\documentclass[a4paper,USenglish,thm-restate]{lipics-v2021}

\nolinenumbers

\usepackage{lipsum}
\usepackage{amsfonts}
\usepackage{graphicx}
\usepackage{epstopdf}
\usepackage{url}
\usepackage{wrapfig}

\ifpdf
  \DeclareGraphicsExtensions{.eps,.pdf,.png,.jpg}
\else
  \DeclareGraphicsExtensions{.eps}
\fi

\title{Frequency-Competitive Query Strategies to Maintain Low Congestion Potential Among Moving Entities\thanks{This work was funded in part by 
Discovery Grants
from the Natural Sciences and Engineering Research Council of Canada.}
}

\titlerunning{Frequency-Competitive Query Strategies for Low Congestion Potential}

\author{William Evans}{Dept. of Computer Science, University of British Columbia, Vancouver, Canada}
{will@cs.ubc.ca} {}{}

\author{David Kirkpatrick}{Dept. of Computer Science, University of British Columbia, Vancouver, Canada}
{kirk@cs.ubc.ca} {}{}

\authorrunning{W. Evans and D. Kirkpatrick}
\Copyright{William Evans and David Kirkpatrick}

\keywords{data in motion, uncertain inputs,
collision avoidance, online algorithms, competitive analysis}
\ccsdesc{Theory of computation~Computational geometry}
\ccsdesc{Theory of computation~Online algorithms}

\usepackage{amsopn}
\usepackage{amsmath,amsthm,amssymb}

\usepackage{changes}
\definechangesauthor[color=purple]{will}
\definechangesauthor[color=red]{kirk}

\usepackage {xspace}
\usepackage{algorithm,caption}
\usepackage[noend]{algpseudocode}
\usepackage{thmtools} 

\newenvironment{scheme}[1][htb]{\begin{algorithm}[#1]%
  }{\end{algorithm}}
\algtext*{Until}

\theoremstyle{claimstyle}

\newcommand {\ply} {\ensuremath{\omega}\xspace}

\newcommand {\mathset} [1] {\ensuremath {\mathbb {#1}}}

\newcommand {\R} {\mathset {R}}

\newcommand{\AAA}{\mathbb{S}}

\newcommand{\PE}{\mathit{PE}}
\newcommand\dif{\mathop{}\!\mathrm{d}}

\usepackage{xifthen}
\newcommand{\NN}[3][]{%
\ifthenelse{\equal{#1}{}}{N_{#3}({#2})}{N_{#3}({#2},{#1})}%
}
\newcommand{\NM}[3][]{%
\ifthenelse{\equal{#1}{}}{M_{#3}({#2})}{N_{#3}({#2},{#1})}%
}
\newcommand{\pNN}[3][]{%
\ifthenelse{\equal{#1}{}}{\widetilde{N}_{#3}({#2})}{N_{#3}({#2},{#1})}%
}
\newcommand{\B}[3][]{%
\ifthenelse{\equal{#1}{}}{B_{#3}({#2})}{B_{#3}({#2},{#1})}%
}
\newcommand{\rd}[3][]{%
\ifthenelse{\equal{#1}{}}{r_{#3}({#2})}{r_{#3}({#2},{#1})}%
}
\newcommand{\prd}[3][]{%
\ifthenelse{\equal{#1}{}}{\widetilde{r}_{#3}({#2})}{\widetilde{r}_{#3}({#2},{#1})}%
}
\newcommand{\sd}[3][]{%
\ifthenelse{\equal{#1}{}}{\sigma_{#3}({#2})}{\sigma_{#3}({#2},{#1})}%
}
\newcommand{\psd}[3][]{%
\ifthenelse{\equal{#1}{}}{\widetilde{\sigma}_{#3}({#2})}{\widetilde{\sigma}_{#3}({#2},{#1})}%
}

\newcommand {\xball}[2][]{%
\ensuremath{%
\ifthenelse{\equal{#1}{}}{\text{ball}({#2})\xspace}{{#1}\text{-ball}({#2})\xspace}%
}}

\newcommand{\Bsec}{B}

\newcommand{\mkmcal}[1]{\ensuremath{\mathcal{#1}}\xspace}

\newcommand{\E}{\mkmcal{E}}

\usepackage{xcolor}

\newcommand{\elide}[1]{}

\date{\vspace{-7ex}}
\begin{document}

\maketitle

\begin {abstract}

Consider a collection of entities moving continuously with bounded speed, but otherwise unpredictably, in some low-dimensional space. Two such entities encroach upon one another at a fixed time if their separation is less than some specified threshold. Encroachment, of concern in many settings such as collision avoidance, may be unavoidable. However, the associated difficulties are compounded if there is uncertainty about the precise location of entities, giving rise to potential encroachment and, more generally, potential congestion within the full collection.

We consider a model in which entities can be queried for their current location (at some cost) and
the uncertainty regions associated with an entity grows in proportion to the time since that entity was last queried. The goal is to maintain low potential congestion, measured in terms of the (dynamic) intersection graph of uncertainty regions,  using the lowest possible query cost.
Previous work [SoCG'13, EuroCG'14, SICOMP'16, SODA'19], in the same uncertainty model, addressed the problem of minimizing the congestion 
potential \emph{of point entities} using location queries of some bounded frequency.
It was shown that it is possible to design a query scheme that is $O(1)$-competitive, in terms of worst-case 
congestion potential, with other, even clairvoyant
query schemes (that know the trajectories of all entities), subject to the same bound on query frequency.

In this paper we address a more general problem with the dual optimization objective:
minimizing the query frequency, measured in terms of total number of queries or the minimum spacing between queries (granularity), over any fixed time interval, 
while guaranteeing a fixed bound on congestion potential \emph{of entities with positive extent}.
This complementary objective necessitates quite different algorithms and analyses. Nevertheless, our results parallel those of the earlier papers, specifically tight competitive bounds on required query frequency, with a few surprising differences.

\end {abstract}

%\newpage

%\thispagestyle{empty}
%\tableofcontents

% \thispagestyle{empty}
% \newpage
% \setcounter{page}{1}

%\category{F.2.2}{Analysis of Algorithms and
%  Problem Complexity}{Non-numerical Algorithms and
%  Problems}[Geometrical problems and computations]
%\category{F.1.2}{Computation by Abstract Devices}{Modes of
%  Computation}[Online computation]

% General Terms (need to be chosen in the submission system, but
% remember this section does not have to appear on the first page)
%\terms{Theory}

%\keywords{input imprecision, competitive analysis, kinetic data}

\section {Introduction}

This paper addresses a fundamental issue in algorithm design, of both theoretical and practical interest: how to cope with unavoidable imprecision in data.
We focus on a class of problems associated with location uncertainty arising from the motion of independent entities when
location queries
to reduce uncertainty are
expensive.
For concreteness, imagine a collection of 
robots following unpredictable trajectories with bounded speed. If individual robots are not monitored continuously there is unavoidable uncertainty, growing with the length of unmonitored activity, concerning their precise location. This portends some risk of collision with neighbouring robots, necessitating some perhaps costly collision avoidance protocol. Nevertheless, robots that are known to be well-separated at some point in time will remain free of collision for the near future. How then should a limited query budget be allocated over time so as to minimize the risk of collisions or, more realistically, help focus collision avoidance measures on robot pairs that are in serious risk of collision?

We adopt a general framework for addressing such problems, essentially the same as the
one studied by Evans et al.~\cite{EKLS16} (which unifies and improves~\cite{EKLS-SoCG13,EKLS-EuroCG14}) and by Busto et al.~\cite{BEK2019}. In this model, entities may be queried at any time in order to reveal their true location, but between queries location uncertainty,
captured by regions surrounding the last known location,
grows. The goal is to understand with what frequency such queries need to be performed, and which entities should be queried, in order to maintain a particular measure of the congestion potential of the entities, formulated in terms of the overlap of their uncertainty regions.
We present query schemes to maintain several measures of congestion potential that,
over every modest-sized time interval, are competitive
in terms of the frequency of their queries, with any scheme that maintains the same measure over that interval alone.

\subsection{The Query Model}
To facilitate comparisons with earlier results, we adopt much of the notation used by Evans et al.~\cite{EKLS16} and Busto et al.~\cite{BEK2019}.
Let $\E$ be a set $\{ e_1, e_2, \dots, e_n \}$ of (mobile) entities. 
Each entity $e_i$ is modelled as a $d$-dimensional closed ball with fixed extent and bounded speed,
whose position (centre location) at any time is specified by the (unknown) function $\zeta_i$
from $[0, \infty)$ (time) to $\mathbb{R}^d$.
(We take the entity radius to be our \emph{unit of distance}, and take the time for an entity moving at maximum speed to move a unit distance to be our 
\emph{unit of time}.)

\begin{wrapfigure}{r}{5cm}
\includegraphics[scale=0.55]{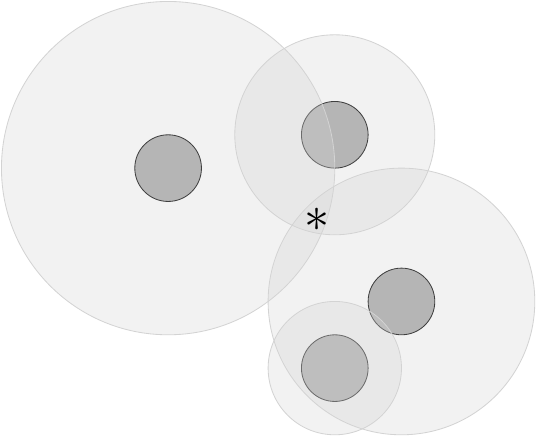}
\caption{Uncertainty regions of four unit-radius entities (dark grey), shown at their most recently known locations, four, three, two and one time unit in the past, respectively. The associated uncertainty ply is three (witnessed by point $*$).}
\label{fig:notations}
\end{wrapfigure}

The $n$-tuple $(\zeta_1(t), \zeta_2(t), \dots, \zeta_n(t))$ is called the \emph{$\E$-configuration} at time $t$.
Entities $e_i$ and $e_j$ are said to \emph{encroach} upon one another at time $t$ if the distance between their centres $\|\zeta_i(t) - \zeta_j(t)\|$ is less than some fixed \emph{encroachment threshold} $\Xi$. 
For simplicity we assume to start that the  distance between entity centres is always at least $2$ 
(i.e. $\|\zeta_i(t) - \zeta_j(t)\| - 2$, the
\emph{separation} between entities $e_i$ and $e_j$ at time $t$, is always at least zero---so entities never properly intersect), and that the encroachment threshold is exactly $2$ (i.e. we are only concerned with avoiding entity contact).
The concluding section considers a relaxation (and decoupling) of these assumptions in which 
$\|\zeta_i(t) - \zeta_j(t)\|$, is always at least some positive constant $\rho_0$ (possibly less than $2$), and 
the encroachment threshold $\Xi$ is some constant
at least $\rho_0$.

We wish to maintain knowledge of the positions of the entities over
time by 
making location queries to individual entities, each of which
returns the exact position of the entity at the time of the query.
A \emph{(query) scheme} $\AAA$ is just an assignment of 
location queries to time instances. 
We measure the \emph{query complexity} of a scheme over a specified time interval $T$ in terms of the total number of queries over $T$ as well as the minimum query \emph{granularity} (the time between consecutive queries) over $T$.

At any time $t\ge 0$, let $p^{\AAA}_i(t)$ denote the time, prior to $t$, that entity $e_i$ was last queried; we define $p^{\AAA}_i(0) = -\infty$. 
The \emph{uncertainty region} of $e_i$ at time $t$, denoted $u^{\AAA}_i(t)$, is defined as the ball with centre $\zeta_i(p^{\AAA}_i(t))$ and radius 
$1 + t - p^{\AAA}_i(t)$ ( cf. Figure~\ref{fig:notations});
note that $u^{\AAA}_i(0)$ is unbounded.
We omit $\AAA$ when it is understood and the dependence on $t$ when $t$ is fixed.

The set $U(t)=\{ u_1(t), \ldots, u_n(t) \}$
is called the \emph{(uncertainty) configuration} at time $t$.
Entity $e_i$ is said to \emph{potentially encroach} upon entity $e_j$
in configuration $U(t)$ if $u_i(t) \cap u_j(t) \neq \emptyset$ (that is, there are potential locations for 
$e_i$ and $e_j$ at time $t$ such that 
$e_i \cap e_j \neq \emptyset$).

In this way any configuration $U$ gives rise to an associated (symmetric) \emph{potential encroachment graph} $\PE^U$ on the set $\E$.  
Note that, by our assumptions above, the potential encroachment graph associated with the initial uncertainty configuration $U(0)$ is complete.

We define the following notions of \emph{congestion potential} 
(called \emph{interference potential} in \cite{BEK2019})
in terms of configuration $U$ and the graph $\PE^U$:

\begin{itemize}
\item the \textbf{(uncertainty) max-degree} (hereafter \textbf{degree}) of the  configuration $U$ is given by $\delta^{U}= \max_i\{\delta^{U}_i\}$
where $\delta^{U}_i$ is defined as the degree of entity $e_i$ in  $\PE^U$ (the maximum number of entities $e_j$, 
\emph{including $e_i$},
that potentially encroach upon $e_i$ in configuration $U$.) 

\item the \textbf{(uncertainty) ply} $\ply^U$ of configuration $U$ is the 
maximum number of uncertainty regions in $U$ that intersect at a single point. This is the largest number of entities in configuration $U$ whose mutual potential encroachment is witnessed by a single point.

\item the \textbf{(uncertainty) thickness} $\chi^U$ of configuration $U$ is the chromatic number of $\PE^U$. This is the size of the smallest decomposition of $\E$ into independent sets (sets with no potential for encroachment) in configuration $U$.
\end{itemize}

Note that $\ply^U \le \chi^U \le \delta^{U}$, so upper bounds on $\delta^{U}$, and lower bounds on $\ply^U$, apply to all three measures.

The assumption that entities never properly intersect is helpful since it means that if uncertainty regions are kept sufficiently small, uncertainty ply can be kept to at most two.
Similarly, for $x$ larger than some dimension-dependent sphere-packing constant, it is always possible to maintain uncertainty degree at most $x$, using sufficiently high query frequency.

\subsection{Related Work}

One of the most widely-studied approaches to computing functions of moving entities uses the \emph{kinetic data structure} model
which assumes precise information about the future trajectories of the
moving entities and relies on elementary geometric relations among their
locations along those trajectories to certify that a
combinatorial structure of interest, such as their convex hull, remains essentially the same.
The algorithm can anticipate when a relation will fail, or is informed
if a trajectory changes, and the goal is
to update the structure efficiently in response to these events~\cite{g-kdssar-98,BGH99,GuibasRoeloffzen17}.
Another less common model assumes that the precise location of
\emph{every} entity is given to the algorithm periodically.  The
goal again is to update the structure efficiently when this
occurs~\cite{deBerg2012quadtrees,deBerg2012hulls,deBerg2013centers}.

More similar to ours is 
the framework introduced by Kahan~\cite{Kahan91,KahanThesis91}
for studying data in motion problems
that require repeated computation of a function (geometric structure)
of data that 
is moving continuously in space where data acquisition
via queries is costly.
There, location queries
occur simultaneously in batches,
triggered by requests (Kahan refers to these as ``queries'') to compute some function at some time rather than by a
requirement to maintain a structure or property at all times.
Performance is compared to a ``lucky'' algorithm that queries the
minimum amount to calculate the function.
Kahan's model and use of competitive evaluation is common to much of
the work on query algorithms for uncertain inputs (see Erlebach and
Hoffmann's survey~\cite{ErlebachSurvey15}).

As mentioned, our model is 
essentially the same as the
one studied by Evans et al.~\cite{EKLS16} 
and by Busto et al.~\cite{BEK2019}, 
both of which focus on point entities.
Paper~\cite{EKLS16} contains strategies whose goal is to guarantee competitively low congestion potential, compared to any other (even clairvoyant) scheme, at one specified target time.
It provides precise descriptions of the impact on this guarantee for several measures of initial location knowledge and available lead time before the target.
The other paper~\cite{BEK2019} contains a scheme for guaranteeing competitively low congestion potential at \emph{all} times.
For this more challenging task the scheme maintains congestion potential over time that is within a constant factor of that maintained by any other scheme over modest-sized time intervals.
All of these results dealt with the \emph{optimization of congestion potential measures subject to fixed query frequency}.

In this paper, we consider 
the dual problem,
for entities of bounded extent:
\emph{optimizing query frequency required to guarantee fixed bounds on congestion}.
This dual optimization is fundamentally different: being able to optimize congestion using fixed query frequency provides no insight into how to optimize query frequency to maintain a fixed bound on congestion. In particular, even for stationary entities, a small change in the congestion bound can lead to an arbitrarily large change in the required query frequency.
Our optimization takes two forms: minimizing total queries (over specified intervals) and then maximizing the minimum query granularity.

\subsection{Our Results}

Our primary goal is to formulate efficient query schemes that, for all possible collections of moving entities, maintain fixed bounds on congestion potential measures at all times. Naturally for many such collections the required query frequency changes over time as entities cluster and spread, so efficient query schemes need to adapt to changes in the configuration of entities. While such changes are continuous and bounded in rate, they are only discernible through  queries to individual entities, so entity configurations are \emph{never} known precisely; future configurations are of course entirely hidden. In this latter respect our schemes and the competitive analysis of their efficiency, using as a benchmark a \emph{clairvoyant scheme} that bases its queries on full knowledge of all entity trajectories (and hence all future configurations), resembles familiar scenarios that arise in the design and analysis of on-line algorithms.

We begin by describing a query scheme to achieve uncertainty degree at most $x$
at one fixed target time in the future 
(reminiscent of the objective in \cite{EKLS16}). 
This serves not only as an interesting contrast (of independent interest) with continuous time optimization, but also plays an important role in the efficient initialization of our continuous schemes.
The detailed description and analysis of our fixed target time scheme, presented in an appendix, show that uncertainty degree at most $x$ can be achieved 
using query granularity that is at most a factor $\Theta(x)$ smaller than that used by
any, even clairvoyant, query scheme to achieve the same goal.
An example shows that this competitive factor is asymptotically optimal. Nevertheless, if we relax our objective, allowing instead 
uncertainty degree at most
$x\!+\!\Delta$,
where $1 \le \Delta \le x$,
the competitive factor $\Theta(x)$ drops to
$\Theta(\frac{x\!+\!\Delta}{1\!+\!\Delta})$.
Again, this competitive factor is shown to be asymptotically optimal
for non-clairvoyant schemes.
This analysis of an algorithm that solves a slightly relaxed optimization, relative to a clairvoyant algorithm that solves the un-relaxed optimization, foreshadows similar analyses of our schemes for continuous time query optimization.

We then turn our attention to 
our primary goal
(the \emph{continuous} case).
At first, we imagine that entities remain stationary, even though they
have the potential to move, since this avoids complications arising
from entities changing location and nearby neighbours.
We show a (tight) factor gap of $\Theta(\ln n)$ between the frequency
required to achieve uncertainty degree (or ply) $x$ at
a fixed-time versus at all times for $n$ stationary entities.
We then define a stationary frequency demand that serves to lower bound the number of queries (and hence upper bound the query granularity) 
required by any, even clairvoyant, scheme 
to maintain uncertainty ply $x$
over any modest-length time interval $T$. 
This is complemented by a scheme that uses $\Theta(x)$ times that number of (well-spaced) queries to achieve uncertainty degree $x$
over \emph{all} such time intervals.
Again the competitive factor drops if 
we permit our scheme to maintain
the relaxed 
degree bound $x\!+\!\Delta$.
Both competitive factors are shown to be best-possible
for non-clairvoyant schemes.

When the entities are mobile, the 
stationary
frequency demand changes over
time as entities cluster and separate
and the analysis becomes significantly more involved.
However, by integrating 
the stationary measure over any modest-length time interval $T$
we can again lower bound the 
number of queries (and hence upper bound the query granularity) 
required by any, even clairvoyant, scheme 
to maintain uncertainty ply $x$ over $T$. 
Again, a query scheme is described that meets these intrinsic bounds to within a factor $\Theta(x)$, over \emph{all} such time intervals $T$. 
As before, 
the competitive factor $\Theta(x)$ drops when we only need to maintain degree $x\!+\!\Delta$, and in all cases it is shown to be best possible.

Not surprisingly, it is possible to further strengthen our competitive bounds by making assumptions about the uniformity of entity distributions. In the concluding discussion,
we consider some ways in which uniformity impacts our results.
We also describe 
several other modifications to our model that make our query optimization framework even more broadly applicable.

\begin{wrapfigure}{r}{5cm}
\includegraphics{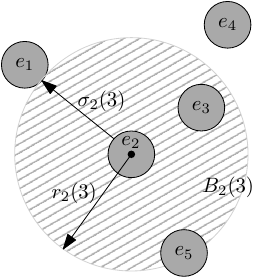}
\caption{A configuration of five unit radius entities. $B_2(3)$, the 3-ball of entity $e_2$, is shown shaded.}
\label{fig:configuration}
\end{wrapfigure}

Returning to the problem concerning collision avoidance mentioned at the start, it follows from our results that, by maintaining uncertainty degree at most $x$, we maintain for each entity $e_i$ a certificate 
identifying
the, at most $x-1$, other entities that could potentially 
encroach upon
$e_i$ (those warranting more careful local monitoring).
An additional application, considered in~\cite{BEK2019}, concerns entities that are mobile transmission
sources, with associated broadcast ranges, where the goal is to minimize the number of broadcast channels so as to eliminate potential transmission interference. 
In this case, maintaining the uncertainty thickness to be at most $x$, using minimum query frequency, serves to maintain a fixed bound on the number of broadcast channels, an objective that seems to be at least as well motivated as optimizing the number of channels for a fixed query frequency (the objective in~\cite{BEK2019}).

\section{Geometric Preliminaries}

In any $\E$-configuration $Z= (z_1, z_2, \ldots, z_n)$ and for any positive integer $x$, we call the separation between $e_i$ and its $x$th closest neighbour (not including $e_i$) its \emph{$x$-separation}, and denote it by $\sigma^Z_i(x)$. 
We call the closed ball with radius (called the \emph{$x$-radius} of $e_i$) 
$r^Z_i(x)=\sigma^Z_i(x) + 1$
and centre $z_i$, the \emph{$x$-ball} of $e_i$, and denote it by $B^Z_i(x)$
(cf. Figure~\ref{fig:configuration}.)
We will omit $Z$ when the configuration is understood.
Note that, for all entities $e_i$ and $e_j$, 
\begin{equation}\label{eqn:observationA}
    \sigma_j^Z(x) \le \| z_j - z_i \| + \sigma_i^Z(x), 
\end{equation}
since the ball with radius 
$\| z_j - z_i \| + \sigma_i^Z(x)$
centred at $z_j$ contains the ball with radius 
$\sigma_i^Z(x)$
centred at $z_i$ (by the triangle inequality).

We first observe that, at any time, individual entities have bounded overlap with the
$\hat{x}$-balls of other entities.
See Appendix~\ref{app:BallCoverNEW} for the proof of the following:
\begin{restatable}{lem}{LemBallCoverNEW}\label{lem:ballcoverNEW}
In any $\E$-configuration $Z$, any entity $e_*$ intersects the $\hat{x}$-balls of at most $5^d \hat{x}$ entities.
\end{restatable}

We have assumed that entities do not properly intersect. 
Define $c_{d,x}$ to be the smallest constant such that a unit-radius $d$-dimensional ball $B$ can have $x$ disjoint unit-radius $d$-dimensional balls (not including itself) with separation from $B$ at most $c_{d,x}$.
Thus, 
$\sigma^Z_i(x) \ge c_{d,x}$
and hence 
\begin{equation}\label{eqn:lambda}
    \frac{1- \lambda_{d,x}}{\lambda_{d,x}} \sigma_i^Z(x) \ge 1 
    \mbox{  and   } 
    r_i^Z(x) \le \sigma_i^Z(x)/\lambda_{d,x}
\end{equation}
where $\lambda_{d,x} = \frac{c_{d,x}}{1 + c_{d,x}}$.
Observe that, for any $\xi \ge 0$, 
$c_{d,x} \ge \xi$ 
provided $x \ge (3+ \xi)^d$, 
since unit-radius balls with separation at most $\xi$ from $B$ must all fit within a ball of radius $3 + \xi$ concentric with $B$.
Thus $1/2 \le \lambda_{d,x} < 1$ if $x \ge 4^d$.

Let $X_d$ be the largest value of $x$ for which $c_{d,x} = 0$ (e.g., $X_2=6$).
Clearly, if $x \le X_d$ there are entity configurations $Z$ with $\sigma^Z_i(x)=0$.
Thus, for such $x$, maintaining uncertainty degree at most $x$ 
might be
impossible in general, for any query scheme.
On the other hand, if $x > X_d$, then $0 < \lambda_{d,x} < 1$. 
Thus:

\mbox{}\\
\indent
{\bf Remark.} 
Hereafter \emph{we will assume that $x$, our bound on congestion potential,
is greater than $X_d$}.
The constants $X_d$ and 
$\lambda_{d,x}$
will play a role in both the formulation and analysis of our query schemes in (arbitrary, but fixed) dimension $d$. If the reader prefers to focus on dimension $2$, then it will be safe to assume hereafter that $x\ge 16$ (and so  $1/2 \le \lambda_{d,x}< 1$).

\section{Query Optimization at a Fixed Target Time}\label{sec:fixedtarget}

Suppose our goal, for a given entity set $\E$, is to optimize queries to guarantee low congestion potential at a fixed target time,
say time $\tau$ in the future,
starting from a state of complete uncertainty of entity locations.
One reason might be to prepare for a computation at that time whose efficiency depends on this low congestion potential~\cite{LofflerSnoeyink10}.
Not surprisingly, this 
can also play a
role in initializing our query scheme that optimizes queries continuously from some point in time onward.

Since we assume unbounded uncertainty regions at the start,
any query scheme 
must query at least a fixed fraction of the entities over the interval $[0, \tau]$, provided $\E$ is sufficiently large
compared to the allowed potential congestion measure.
Hence the minimum query granularity over this interval must be $O(\tau/|\E|)$. 
Furthermore, if granularity is not an issue, $O(|\E|)$ queries suffice, provided they are made sufficiently close to the target time. 
Optimizing the minimum query granularity is less straightforward. 
Nevertheless, it is clear that any sensible query scheme, using minimum query granularity $\gamma$, that 
guarantees a given  measure of congestion potential at most $x$ at 
the target time $\tau$, determines a radius 
$1 + k_i \gamma$, 
for each entity $e_i \in \E$, where (i) $k_1, k_2, \ldots k_n$ is a permutation of $1, 2, \ldots n$, and 
(ii) the uncertainty configuration, in which entity $e_i$ has an uncertainty region $u_i$ with centre $\zeta_i(\tau)$ and radius $1 + k_i \gamma$, has the given congestion measure at most $x$.
For any measure, we associate with $\E$ an {\em intrinsic fixed-target granularity}, denoted $\gamma_{\E,x},$ defined to be the largest $\gamma$ for which these conditions are satisfiable.

It is not hard to see that, by projecting the current uncertainty regions to the target time (assuming no further queries), some entities can be declared ``safe'' (meaning their projected uncertainty regions cannot possibly contribute to ply greater than $x$ at the target time). 
This idea is exploited in a query scheme that queries entities in rounds of geometrically decreasing duration, following each of which a subset of ``safe'' entities are set aside with no further attention, until no ``unsafe'' entities remain.
The details and analysis of this
Fixed-Target-Time (FTT) query scheme are presented in Appendix~\ref{app:FTT}. 
Our analysis is summarized in:

\begin{restatable}{theorem}{FixedTargetTimeTheorem}
\label{thm:FTT}
For any $\Delta$,
$0 \le \Delta \le x$,
the FTT$[x\!+\!\Delta]$ query scheme guarantees uncertainty degree at most $x\!+\!\Delta$ at target time $\tau$ and uses minimum query granularity over the interval $[0, \tau]$ that is at most a factor
$\Theta(\frac{x\!+\!\Delta }{ 1\!+\!\Delta})$ smaller than
$\gamma_{\E,x}$.
\end{restatable}

A similar, though somewhat simpler, argument can be used to show that 
a competitive ratio of $\Theta(\frac{x\!+\!\Delta }{ 1\!+\!\Delta})$ on the 
granularity for a scheme that achieves \emph{ply} at most $x\!+\!\Delta$, relative to any, even clairvoyant, scheme that guarantees ply at most $x$, can also be achieved.
It is interesting to see that the 
$O(\frac{x\!+\!\Delta}{1\!+\!\Delta})$ competitive ratio on query granularity for achieving degree at most $x\!+\!\Delta$ (as well as that for achieving ply at most $x\!+\!\Delta$) at a fixed target time cannot be improved in general. 
In fact, we show in Appendix~\ref{ex:MobilePlyFreqX} an example 
demonstrating that, for $1 \le \Delta \le x$
degree at most $x\!+\!\Delta$ can be guaranteed at a fixed target time by a scheme that uses query granularity one, yet any
non-clairvoyant scheme
that guarantees ply at most $x$ at the target time must use at least 
$x+2\Delta +1$ queries over the last $2(1\!+\!\Delta)$ steps.

\section{Continuous Query Optimization: Stationary Entities}

For stationary entities, maintaining a fixed configuration $Z$ over time,
uncertainty degree at most $x$ can be maintained on a continuous basis using granularity at least $\Theta(\frac{1}{\ln n})$ times 
what we called the intrinsic fixed-target granularity ($\gamma_{\E,x}$), required to achieve uncertainty degree at most $x$ at any fixed time. %(even if in the latter all locations are known at the start).
This follows because the uncertainty region of each entity $e_i$ can be kept, at all times, within its corresponding fixed-target realization,
simply 
by querying it at least once every 
$k_i \gamma_{\E,x}$ time steps. By Corollary 3.2 of \cite{Holte89}, we know that such a schedule is feasible using query granularity $\gamma_{\E,x} / (2 \ln n)$.
A simple example (see Appendix~\ref{ex:stationaryPlyFreqX}) shows that this  $\Theta(\frac{1}{\ln n})$ gap
is unavoidable in general.

\subsection{Intrinsic Frequency Demand: Stationary Entities}

If the radius of the uncertainty region $u_i$ of $e_i \in \E$ exceeds 
$1 + \sigma_i(x)$
then $u_i$ intersects at least $x+1$ entities in $\E$. 
Thus $\sigma_i(x)$, the $x$-separation of $e_i$,  is an upper bound on the amount of uncertainty whose avoidance guarantees
that the uncertainty degree of entity $e_i$ remains at most $x$. 
It follows that, 
for stationary entities,  each entity $e_i$ must be queried with frequency at least $1 / \sigma_i(x)$.
Hence $\phi_{\E, x} = \sum_{e_i \in \E} \frac{1}{\sigma_i(x)}$,
the \emph{stationary frequency demand},
provides a lower bound on the 
total query frequency {(measured as queries per unit of time)}\footnote{This frequency demand appears in work on scheduling jobs according to a vector $v$ of periods, where job $j$ must be scheduled at least once in every time interval of length $v_j$~\cite{Holte89,Fishburn2002}.} 
required to avoid uncertainty degree greater than $x$, with entities in stationary configuration $Z$.
Note that, since we have assumed that $x > X_d$, we know that
$\sigma_i(x)$ is some strictly positive multiple (at least $\lambda_{d,x}$) of $r_i(x)$
(see Equation~(\ref{eqn:lambda})).
Hence, 
$\lambda_{d,x} \phi_{\E, x} \le \sum_{e_i \in \E} \frac{1}{r_i(x)}$.

The following lemma shows that an expression closely related to $\phi_{\E, x}$ 
plays an essential role as well in bounding
the total number of queries required to maintain uncertainty ply at most $x$ over a specified time interval. Note that it applies to strategies that maintain ply at most $x$ \emph{only} over $T$, \emph{even if they allow ply greater than $x$ elsewhere}.

\begin{restatable}{lem}{LemmaStationaryPlyLB}
\label{newnewlem:stationaryPlyLB}
Let 
$0 \le \Delta \le x$
and let $T$ be a time interval for which 
$|T| \phi_{\E, x\!+\!\Delta} \ge 2|\E|$.
Then any query scheme with maximum uncertainty ply at most $x$ over $T$ 
must make a total of at least 
$\frac{(\Delta+1) \lambda_{d,x\!+\!\Delta}}{ 4 \cdot 5^d (x\!+\!\Delta)} \cdot 
|T| \phi_{\E, x\!+\!\Delta}$
queries over $T$.
\end{restatable}

\begin{proof}
Let $\E' = \{ e_i \in \E \mid r_i(x\!+\!\Delta) \le |T| \}$. 
Then $|T| \phi_{\E', x\!+\!\Delta} \ge |T| \phi_{\E, x\!+\!\Delta} - |\E|
\ge |T| \phi_{\E, x\!+\!\Delta}/2$. 

For each  $e_i \in \E'$ partition $T$ into 
$\big\lfloor \frac{|T|}{r_i(x\!+\!\Delta)}\big\rfloor 
\ge \frac{|T|}{2 r_i(x\!+\!\Delta)}$ sub-intervals each of length at least $r_i(x\!+\!\Delta)$. 
Let $E_i$ denote the set of (at least) $x\!+\!\Delta$ entities in $\E \setminus e_i$ that intersect the 
$(x\!+\!\Delta)$-ball of entity $e_i$.
We say that entity $e_i$ is \emph{satisfied} in a sub-interval if at least $\Delta + 1$ of the
entities in $E_i$ are queried in that sub-interval. 

If $e_i$ is not satisfied in a given sub-interval then at the end of the sub-interval at least $x$ of the $x\!+\!\Delta$ entities in $E_i$ must have uncertainty regions of radius at least $r_i(x\!+\!\Delta)$, and hence intersect $z_i$, forming ply at least $x+1$ at that point. 
Thus, to avoid ply greater than $x$ throughout $T$, every entity $e_i$ must be satisfied in each of its sub-intervals, and so have associated with it at least $(\Delta +1) \cdot  \frac{|T|}{2 r_i(x\!+\!\Delta)} $ queries. 

But, by Lemma~\ref{lem:ballcoverNEW} entity $e_j$ can intersect
the $(x\!+\!\Delta)$-ball of at most $5^d (x\!+\!\Delta)$ entities, and so a query to $e_j$ can serve to satisfy at most $5^d (x\!+\!\Delta)$ other entities. 
It follows that the total number of queries needed to avoid ply greater than $x$ throughout interval $T$ is at least
\begin{multline*} 
(\Delta +1) \cdot \sum_{e_i \in \E'}  \frac{|T|}{2 r_i(x\!+\!\Delta)}  \frac{1}{\cdot 5^d (x\!+\!\Delta)} 
\ge  \frac{(\Delta+1) }{ 2\cdot 5^d (x\!+\!\Delta)} |T| \sum_{e_i \in \E'}  \frac{1}{r_i(x\!+\!\Delta)} \\
\ge  \frac{(\Delta+1) }{2 \cdot 5^d (x\!+\!\Delta)} \cdot 
|T| \lambda_{d,x\!+\!\Delta} \phi_{\E', x\!+\!\Delta}
\ge   \frac{(\Delta+1) \lambda_{d,x\!+\!\Delta}}{ 4 \cdot 5^d (x\!+\!\Delta)} \cdot 
|T| \phi_{\E, x\!+\!\Delta}  
.
\end{multline*}
\end{proof}

We next
show that 
it is possible to maintain uncertainty degree (and hence uncertainty ply) at most $x$ using a query scheme with minimum query granularity that is $\Theta(1/ \phi_{\E, x})$.

\subsection{A Scheme to Maintain Low Degree: Stationary Entities}

The \emph{Frequency-Weighted Round-Robin} 
scheme for maintaining uncertainty degree at most $x$, denoted ${\rm FWRR}[x]$,  queries according to a schedule in which, for all $i$,
entity $e_i$ is queried once every 
$t_i = 2^{g + \lfloor \lg  [ \sigma_i(x) \lambda_{d,x} / (\lambda_{d,x} +2)] \rfloor}$
time steps of size (granularity) $1/2^g$, 
where 
$g = \lceil \lg( \frac{\lambda_{d,x} +2}{\lambda_{d,x}} \phi_{\E, x}) \rceil +1$.
The schedule repeats after $\max\{t_i\}$ steps.

\begin{lemma}
\label{lem:FWRRupper}
The query scheme  ${\rm FWRR}[x]$ maintains uncertainty configurations with uncertainty degree at most $x$ 
at all times,  and has an implementation using minimum query granularity at least
$\frac{\lambda_{d,x}}{4(\lambda_{d,x} +2)} \frac{1}{ \phi_{\E, x}}$.
\end{lemma}

\begin{proof}
Since $t_i/2^g \le 
\frac{\lambda_{d,x}}{\lambda_{d,x} + 2} \sigma_i(x)$,
it follows that 
${\rm FWRR}[x]$ will query every entity $e_i$ with at most
$\frac{\lambda_{d,x}}{\lambda_{d,x} + 2} \sigma_i(x)$ time between queries. 
Equation~(\ref{eqn:observationA}) implies that any entity $e_j$ whose separation from $e_i$ is $s \ge \sigma_i(x)$ 
has the property that $\sigma_j(x) \le s + 2 + \sigma_i(x) \le 2s + 2$. 
Hence, using equation~(\ref{eqn:lambda}),
\[ 
\sigma_i(x) + \sigma_j(x) \le 3s + 2  \le (3 + 2\frac{1- \lambda_{d,x}}{\lambda_{d,x}}) s = \frac{\lambda_{d,x} + 2}{\lambda_{d,x}} s,
\]
and 
$\frac{\lambda_{d,x}}{\lambda_{d,x} + 2} \sigma_i(x) +
\frac{\lambda_{d,x}}{\lambda_{d,x} + 2} \sigma_j(x)  \le s$.
So the uncertainty regions of $e_i$ and $e_j$ never properly intersect,
and thus the uncertainty degree of $e_i$ remains at most $x$ over time.

Since $\sum_{e_i \in \E} \frac{1}{t_i} < \frac{\lambda_{d,x} + 2}{\lambda_{d,x} 2^{g-1}} \phi_{\E, x} \le 1$,
it follows from a result of Anily et al.~\cite{anily98}
(see Lemma~6.2) that a query schedule exists with at most one query for every slot of size $1/2^g$.
Hence, the FWRR query scheme can be implemented with query granularity at least 
$\frac{\lambda_{d,x}}{4(\lambda_{d,x} +2)} \frac{1}{ 
\phi_{\E, x}
}$.
\end{proof}

Combining Lemma~\ref{newnewlem:stationaryPlyLB} with Lemma~\ref{lem:FWRRupper} we get:

\begin{theorem}
\label{thm:FWRRplycompetitive}
For any $\Delta$,
$0 \le \Delta \le x$,
the ${\rm FWRR}[x\!+\!\Delta]$ query scheme  maintains uncertainty configurations with 
uncertainty 
degree 
at most $x\!+\!\Delta$
using minimum query granularity at least
$\frac{\lambda_{d,x}}{4(\lambda_{d,x} +2)} \frac{1}{ \phi_{\E, x\!+\!\Delta}}$.
Furthermore, it uses a total number of queries over any 
time interval $T$
satisfying  
$|T| \phi_{\E, x\!+\!\Delta} \ge 2|\E|$,
that is competitive
with any scheme $\mathbb{S}$ that maintains uncertainty configurations with uncertainty ply at most $x$ over $T$, even if $\mathbb{S}$ permits uncertainty ply larger than $x$ elsewhere. 
The competitive factor is
$O(\frac{x\!+\!\Delta}{1\!+\!\Delta})$.
\end{theorem}

Appendix~\ref{ex:stationaryPlyFreqX} describes a collection of stationary entities 
showing that, for any 
$\Delta$,
$0 \le \Delta \le x$,
ply at most $x$ can be maintained with query granularity one, yet any scheme that guarantees uncertainty degree at most $x\!+\!\Delta$ must use at least 
$x\!+\!\Delta + 1$ queries over any time interval of length 
$4(1\!+\!\Delta)$.
Thus the worst-case competitiveness of the ${\rm FWRR}[x\!+\!\Delta]$ query scheme is asymptotically optimal not only for maintaining low uncertainty \emph{degree}, it is also optimal among all strategies that maintain low uncertainty \emph{ply} by maintaining low uncertainty degree.

\section{Continuous Query Optimization: General Mobile Entities}

While the case of stationary entities exhibits some of the difficulties in
maintaining uncertainty regions with low congestion, mobile entities add an
additional level of complexity.

Since an $\E$-configuration may now change over time, we add a parameter $t$ to our stationary definitions, and refer to $x_\E(t)$, $B_i(x,t)$, $\sigma_i(x,t)$, and $r_i(x,t)$
in place of their stationary counterparts at time $t$, where it is understood that the configuration in question is just the $\E$-configuration at time $t$.

As we have seen, when entities are stationary, the expression 
$|T| \phi_{\E, x} = \sum_{e_i \in \E} \frac{|T|}{\sigma_i(x)}$,  
the {\em stationary frequency demand} over time interval $T$,
plays a central role in characterizing the unavoidable number of queries needed to avoid uncertainty degree (or uncertainty ply) greater than $x$. 
For mobile entities, we will refer to the more general expression
$\phi_{\E,x}(T)=\sum_{e_i \in \E} \int_{T} \frac{\dif{t}}{\sigma_i(x, t)}$
as 
the {\em intrinsic frequency demand over time interval $T$},
with no risk of confusion.

\subsection{Intrinsic Frequency Demand: General Mobile Entities}

It follows from earlier work on the optimization of query degree using fixed query frequency~\cite{BEK2019} that a high stationary frequency demand {\em at one instant in time} does not necessarily imply that uncertainty degree at most $x$ is unsustainable.
Nevertheless, as the following lemma demonstrates, high intrinsic frequency demand, 
over a sufficiently large time interval $T$, does imply a lower bound on the number of queries over $T$.

\begin{restatable}{lem}{MobilePlyLBLemma}
\label{lem:mobilePlyLB}
Let $0 \le \Delta \le x$
and let $T$ be a time interval for which
$\phi_{\E,x\!+\!\Delta}(T) \ge 505 |\E|$.
Define $T^+$ to be the interval $T$ extended by $|T|/3^{55}$.
Then any query scheme with maximum uncertainty ply at most $x$ over $T^+$ 
must make a total of at least 
$\frac{(\Delta +1) \lambda_{d,x\!+\!\Delta}}{75 [(99 (354)^d + (489)^d] (x\!+\!\Delta)}
\phi_{\E,x\!+\!\Delta}(T)$
queries over $T^+$.
\end{restatable}

\begin{proof}
(See Appendix~\ref{sec:MobilePlyLBLemma} for full proof.)
%\david{need to modify this sketch...}
At a high level the proof parallels that of Lemma~\ref{newnewlem:stationaryPlyLB}.
However, in the mobile case the
$
(x\!+\!\Delta)$-radius, and indeed the $
(x\!+\!\Delta)$-neighbourhood
of each entity, changes over time.
A reasonable hope is that the \emph{integral} of the entity's inverse $
(x\!+\!\Delta)$-radius over
$T$,
summed over all entities, provides a similar basis for a lower bound.
Certainly, each entity requires 
all but $x-1$ of
its $(x\!+\!\Delta)$-neighbours
to be queried to avoid ply $x$ within a sub-interval
of $T$ of length proportional to its $x$-radius.
The difficulty is that one mobile entity can be the
$(x\!+\!\Delta)$-neighbour of
many entities over time so that one query to that entity can help satisfy the
demands of many sub-intervals.
(In the stationary case, we saw that one query can help satisfy the demands of at most 
$\Theta(x\!+\!\Delta)$
entities since this is the maximum number of stationary
$(x\!+\!\Delta)$-neighbourhoods a stationary entity can be in.)
However, if we restrict our attention to sub-intervals of an entity
during which the entity's
$(x\!+\!\Delta)$-radius remains approximately the same
size, we can apply something similar to the stationary case argument.
The challenge is to show that such sub-intervals, that are not simultaneously partially satisfied together with a large number of other sub-intervals, cover a substantial
fraction of $T$ for many entities.
\end{proof}

It turns out that the lower bound implicit in Lemma~\ref{lem:mobilePlyLB} holds for the interval $T$ itself, or a very small shift of $T$: 

\begin{restatable}{cor}{MobilePlyLBCor}
\label{cor:mobilePlyLB}
Let $0 \le \Delta \le x$
and let $T$ be a time interval for which
$\phi_{\E,x\!+\!\Delta}(T) \ge 1010 |\E|$.
Define 
$\overrightarrow{T}$ to be the interval $T$ shifted by $|T|/3^{55}$
and $T^+ = T \cup  \overrightarrow{T}$.
Then any query scheme with maximum uncertainty ply at most $x$ over $T^+$ 
must make a total of at least 
$\frac{(\Delta +1) \lambda_{d,x\!+\!\Delta}}{75 [(99 (354)^d + (489)^d] (x\!+\!\Delta)}
\phi_{\E,x\!+\!\Delta}(T)$
queries over either $T$ or $\overrightarrow{T}$.
\end{restatable}
\begin{proof}
See Appendix~\ref{sec:ProofMobilePlyLBCor}.
\end{proof}

\subsection{Perception Versus Reality}
For any query scheme,
the true location of a moving entity $e_i$ at time $t$, $\zeta_i(t)$,  may 
differ from its \emph{perceived location}, $\zeta_i(p_i(t))$, its location at the time of its most recent query.
Let $\pNN[t]{x}{i}$ be $e_i$ plus the set of $x$ entities whose perceived locations at time $t$ are closest to the perceived location of $e_i$ at time $t$.
The \emph{perceived $x$-separation} of $e_i$ at time $t$, denoted $\psd[t]{x}{i}$, is the separation between $e_i$ and its perceived $x$th-nearest-neighbour at time $t$, i.e., $\psd[t]{x}{i} = \max_{e_j \in \pNN[t]{x}{i}} \| \zeta_i(p_i(t)) - \zeta_j(p_j(t)) \| - 2$.
The \emph{perceived $x$-radius} of $e_i$ at time $t$, denoted $\prd[t]{x}{i}$, is just $1 + \psd[t]{x}{i}$.

Since a scheme only knows the perceived locations of the entities, it is important that each entity $e_i$ be probed sufficiently often that its 
perceived $x$-separation $\psd[t]{x}{i}$ closely approximates its true $x$-separation $\sd[t]{x}{i}$
at all times $t$.
The following technical lemma
asserts that 
once a close relationship between perception and reality has been established, it
can be sustained by ensuring that the time between queries to an entity is bounded by
some small fraction of 
its perceived $x$-separation. See Appendix~\ref{sec:NEWPerceptionLemma} for the proof.

\begin{restatable}{lem}{NEWNEWLemmaPerceptionPercQuery}
\label{lem:NEWPerception}
Suppose that for some $t_0$ and
for all 
entities $e_i$,
\begin{enumerate}
\item[(i)] 
$\sd[p_i(t_0)]{x}{i}/2 \leq \psd[p_i(t_0)]{x}{i} \leq  3\sd[p_i(t_0)]{x}{i}/2$, 
{\rm [perception is close to reality for $e_i$  at time $p_i(t_0)$]}
and
\item[(ii)] for any $t \ge t_0$, 
$t - p_i(t) \leq \lambda_{d,x} \psd[p_i(t)]{x}{i} / 12 $
{\rm [all queries are done promptly based on perception]}.
\end{enumerate}
Then 
for all entities $e_i$, 
$\sd[t]{x}{i} /2 \leq \psd[t]{x}{i} \leq  3 \sd[t]{x}{i}/2$,
for all $t \geq p_i(t_0)$.
\end{restatable}

To obtain the preconditions of Lemma~\ref{lem:NEWPerception}, we could assume that 
all entities are queried very quickly using low granularity for a short initialization phase.
In Appendix~\ref{app:init}, we show how to use a modified version of the FTT scheme of Section~\ref{sec:fixedtarget} to obtain these preconditions using granularity that is competitive with any scheme that guarantees uncertainty degree at most $x$ from time $t_0$ onward.
This establishes:

\begin{restatable}{lem}{LemInit}
\label{lem:init}
For any
$\Delta$, $0 \le \Delta \le x$,
and any target time $t_0 \geq 0$,
there exists an initialization scheme 
that guarantees 
\begin{enumerate}
\item[(i)] $\sd[t_0]{x\!+\!\Delta}{i}/2 \le \psd[t_0]{x\!+\!\Delta}{i} \leq  3\sd[t_0]{x\!+\!\Delta}{i}/2$,
and
\item[(ii)] $t_0 - p_i(t_0) \leq \lambda_{d,x} \psd[p_i(t_0)]{x\!+\!\Delta}{i} / 12 $.
\end{enumerate}
using minimum query granularity over the interval $[0, t_0]$ that is at most $\Theta(\frac{x\!+\!\Delta}{1\!+\!\Delta})$ smaller than the minimum query granularity, over the interval $[0, (a+1)t_0]$,  used by any other scheme that guarantees uncertainty degree at most $x$ in the interval 
$[t_0, (a+1)t_0]$,
where $a = 64/(5\lambda_{d,x})$.
\end{restatable}

\subsection{A Scheme to Maintain Low Degree: General Mobile Entities}

A \emph{bucket} is a 
set of entities and an associated 
time interval
whose length 
(the bucket's \emph{length}) 
is a power of two.
The 
$i$th bucket $B$ of
length $2^b$
has time interval
$T_B = [i2^b, (i+1) 2^b)$, for integers $i$ and $b$.
The time intervals 
of buckets of the same length
partition $[0,\infty)$, and a bucket of length $2^b$ spans exactly $2^s$ \emph{sub-buckets} of length $2^{b-s}$.

Entities are assigned to exactly one bucket at any moment in time. 
Membership of entity $e_j$ in a given bucket $B$ 
implies a commitment to query $e_j$ within the interval $T_B$. 
The \emph{basic} version of the BucketScheme (see Alg.~\ref{str:bucket}) fulfills these commitments by scheduling a query to $e_j$ at anytime within that time interval.
That is, any version of Schedule($e_j$, $B$) that allocates a query for $e_j$ at some time within 
$T_B$ satisfies the basic BucketScheme.
After an entity $e_j$ is queried, it is reassigned to a future bucket in a way that preserves (via Lemma~\ref{lem:NEWPerception}) the following invariants:
for all
$t' \in T_B$, 
(i) $\sd[t']{x}{j} /2 \leq \psd[t']{x}{j} \leq  3 \sd[t']{x}{j}/2$; and
(ii) $\sd[t']{x}{j} = \Theta(2^b)$, so $\int_{T_B} \frac{\dif{t}}{\sigma_i(x, t)} = \Theta(1)$.

\begin{scheme}
\caption{\textbf{BucketScheme$[x]$}}\label{str:bucket}
\begin{algorithmic}[1]
\State Assume perception-reality precondition properties hold at time $t_0$. \Comment{See Lemma~\ref{lem:init}}
\ForAll{entities $e_j$} \Comment{make initial query-time assignments}
\State Assign $e_j$ to the first bucket $B$ 
of length $2^b$ starting after time $t_0$, 
where 
$b = \lfloor \lg [(\lambda_{d,x} /24) \psd[t_0]{x}{j}] \rfloor$
\State Schedule($e_j$, $B$) \Comment{Assign $e_j$ a query time in interval of bucket $B$}
\EndFor
\Repeat 
\State Advance $t$ to the next query time (say to entity $e_j$)
\State Query $e_j$ 
\State Assign $e_j$ to the next bucket $B$ 
of length $2^b$ starting after time $t + 2^b$, 
where 
$b = \lfloor \lg [(\lambda_{d,x} /24) \psd[t]{x}{j}] \rfloor$
\State Schedule($e_j$, $B$) \Comment{Assign $e_j$ a query time in interval of bucket $B$}
\Until{}
\end{algorithmic}
\end{scheme}

\begin{restatable}{theorem}{ThmBAsuccess}\label{thm:BAsuccess}
The basic BucketScheme$[x]$ maintains uncertainty degree at most $x$ indefinitely. 
Furthermore, 
over any 
time interval $T$ in which the basic BucketScheme$[x]$ makes 
$3|\E|$ queries, 
$\phi_{\E,x}(T) = \Omega(|\E|) $.
\end{restatable}

\begin{proof}
(See Appendix~\ref{sec:basicbucket} for a full proof.)
It is straightforward to confirm that the assignment of entities to buckets (specified in line 7) ensures that the time between successive queries to any entity $e_i$ satisfies precondition (ii) of Lemma~\ref{lem:NEWPerception},
and that new bucket assignments are disjoint from previous bucket assignments.
From the proof of Lemma~\ref{lem:NEWPerception} we see that this in turn implies that 
$t - p_i(t) \le \frac{\lambda_{d, x}}{\lambda_{d, x} +2} \sigma_i(x, t)$, for all entities $e_i$ and all $t \ge t_0$. 
Hence, following the identical analysis used in the proof of Lemma~\ref{lem:FWRRupper}, we conclude that uncertainty degree at most $x$ is maintained indefinitely.

If BucketScheme$[x]$ makes $3|\E|$ queries over $T$ then, among these, it must make at least $|\E|$ queries to entities in buckets that are fully spanned by $T$. 
Since each entity in each fully spanned bucket contributes $\Theta(1)$ to $\phi_{\E,x}(T)$ (and the buckets occupied by any one entity over time are disjoint),
it follows that $\phi_{\E,x}(T) = \Omega(|\E|)$.
\end{proof}

A more fully specified implementation of the BucketScheme is not only competitive in terms of total queries over reasonably small intervals, but also competitive in terms of query granularity. 
The idea of this \emph{refined} BucketScheme
is to 
replace the simple scheduling policy Schedule of the basic BucketScheme
with a recursive policy Schedule* that generates a 
refined reassignment of entities to buckets. 
Whenever a bucket $B$ of length $2^{b}$ has been assigned two entities, these entities are immediately reassigned, one to each of the two sub-buckets of $B$ of length $2^{b-1}$. 
In this way, when all reassignments are finished, all of the entities are assigned to their own
buckets.
The entity associated with a  bucket $B$ has a \emph{tentative next query time} at the midpoint of $B$. 
Tentative query times are updated of course when entities are reassigned (see Alg.~\ref{alg:schedule}). 
At any point in time the next query is made to the entity with the earliest associated tentative next query time.
Note that since distinct buckets have distinct midpoints, and no bucket has more than one associated entity, the current set of tentative next query times contains no duplicates.
In fact, for any two tentative query times associated with entities in buckets $B$ and $B'$, it must be that either $B$ and $B'$ are disjoint, or the smaller bucket is a sub-bucket of one half of the larger bucket.

Recall from the invariant properties of bucket assignments in the basic BucketScheme
that if $e_i$ is assigned to bucket $B$, then 
$\sigma_i(x, t) = \Theta(|T_B|)$
(i.e. $1/ \sigma_i(x, t) = \Theta(1/|T_B|)$), for $t \in T_B$.
In the refined BucketScheme this property is generalized to:
(i) if $e_i$ is assigned to bucket $B$, then there is a subset of entities $S_B$, including $e_i$, such that 
$\sum_{e_j \in S_B} 1/ \sigma_j(x, t) = \Theta(1/|T_B|)$, for $t \in T_B$, and
(ii) if $T_B \cap T_{B'} \neq \emptyset$ then $S_B \cup S_{B'} = \emptyset$.
(It is straightforward to confirm that this property is preserved by the reassignment of entities in the 
bucket structure.)

Since the gap between successive queries contains half of the smaller of the two
buckets containing the two entities, it follows that every gap between queries has an associated integral of $\sum_{e_j \in \E} 1/ \sigma_j(x,t)$ that is $\Theta(1)$.  
It follows from this that the stationary frequency demand  is inversely proportional to the instantaneous granularity at the time of every query.

\begin{scheme}
\caption{\textbf{Schedule*$(e_j, \Bsec)$} \hfill $\triangleright$ used by the refined BucketScheme}\label{alg:schedule}
\begin{algorithmic}[1]
\If{bucket $\Bsec$ already contains an entity $e_i$} \Comment{$\Bsec$ contains at most one}
  \State Unassign $e_i$ from $B$
  \State Schedule*$(e_i, \Bsec_{\text{first}})$
  \Comment{$\Bsec_{\text{first}}$ spans the first half-interval of $\Bsec$}
  \State Schedule*$(e_j, \Bsec_{\text{second}})$
  \Comment{$\Bsec_{\text{second}}$ spans the second half-interval of $\Bsec$}
\Else
  \State Assign $e_j$ to bucket $\Bsec$ with query time
  at the midpoint of $\Bsec$.
\EndIf
\end{algorithmic}
\end{scheme}

We summarize with:

\begin{lemma}\label{lem:totalqueryanalysis2}
Over any 
time interval $T$ in which the refined BucketScheme$[x]$ makes 
$3|\E|$ queries, 
$\phi_{\E,x}(T) = \Omega(|\E|) $.
Furthermore, at any time the query granularity is inversely proportional to the stationary query demand.
\end{lemma}

Combining 
Lemma~\ref{lem:totalqueryanalysis2}
and Corollary~\ref{cor:mobilePlyLB}, 
we reach our main result:

\begin{restatable}{theorem}{ThmFullAdaptive}\label{thm:fulladaptive}
For any $\Delta$,
$0 \le \Delta \le x$,
the refined BucketScheme$[x\!+\!\Delta]$ 
maintains uncertainty degree at most $x\!+\!\Delta$ and, over all sufficiently large time intervals $T$, is competitive, in terms of total queries over $T$ or some small shift $\overrightarrow{T}$ of $T$, with any query scheme that maintains uncertainty ply at most $x$ over $T \cup \overrightarrow{T}$.
The competitive factor is
$O(\frac{x\!+\!\Delta}{1\!+\!\Delta})$.
Furthermore, at all times it uses query granularity that is inversely proportional to the stationary frequency demand.
\end{restatable}

Appendix~\ref{ex:MobilePlyFreqX} describes a collection of mobile entities
for which every query scheme that maintains uncertainty
ply at most $x$
at all times 
in a specified time interval
needs to use 
a query total that is at most a factor 
$\frac{x+2\Delta +1}{2\Delta + 3}$ smaller 
(and hence 
minimum query granularity that is
at most a factor $\frac{2\Delta + 3}{x+2\Delta +1}$ larger)
over the full interval
than 
that used by the best query scheme for achieving uncertainty
degree at most $x\!+\!\Delta $
at all times in the same interval.
It follows that the competitive factor on total queries realized by the 
refined BucketScheme, for maintaining uncertainty degree at most $x\!+\!\Delta$ 
(or for maintaining uncertainty ply at most $x\!+\!\Delta$ by maintaining uncertainty degree at most $x\!+\!\Delta$) cannot be improved by more than a constant factor in general.

\section{Discussion}

\subsection{Motivating applications revisited}
We return briefly to the motivating applications mentioned in the introduction. 
For the collision avoidance application, recall that by maintaining uncertainty degree at most $x$, we maintain (using optimal query frequency) for each robot (entity) $e_i$ a certificate
identifying the, at most $x-1$, other robots that could potentially collide with $e_i$ (those warranting more careful local monitoring).
As described in the next sub-section, it is straightforward to make our results
even more directly useful, in this and other applications, by strengthening the notion of encroachment to hold when the encroachment threshold is any positive constant.

For the channel assignment application, observe that when our scheme queries a transmission source $e_i$ (with associated broadcast range), it schedules the next query of $e_i$ knowing the set $E$ of entities, other than $e_i$ itself, whose broadcast ranges may potentially conflict with that of $e_i$ from now until that query. 
If $e_i$ is assigned a broadcast channel that differs from the broadcast channels assigned to entities in $E$, this requires at most $|E|+1$ channels, which if the scheme maintains uncertainty degree at most $x$, is at most $x$. (Note that channel assignments are updated \emph{locally}, i.e., only the assignment of the just-queried entity changes.)
Our scheme guarantees uncertainty degree $x\!+\!\Delta$ using a query frequency that is 
(up to a constant factor) optimally competitive with that required of any scheme to maintain uncertainty ply (which bounds from above the number of broadcast channels used to avoid potential broadcast interference) at most $x$. 
As we describe in the next subsection, our assumption of disjoint entities (i.e. broadcast ranges) is easily relaxed to permit intersections as long as the broadcast centres remain separated by at least some fixed positive distance.

\subsection{Generalizations of our Model and Analysis}

We describe below several modifications to our model and analyses that make
our query optimization framework more broadly applicable. 

\paragraph*{Relaxing the assumption on the encroachment threshold and entity disjointness }

Without changing the units of distance and time, we can model a collection of unit-radius entities  
any pair of which possibly intersect, but whose centres always maintain distance at least some 
positive constant $\rho_0 < 2$, by simply 
scaling the constant $c_{d, x}$ by $\rho_0/2$ (and the constant $\lambda_{d, x}$
accordingly.
Similarly (and simultaneously), we can model a collection of unit-radius entities with encroachment threshold $\Xi > 2$ by (i) changing the \emph{basic uncertainty radius} (the radius of the uncertainty region of an entity immediately after it has been queried) to $\Xi/2$ (thereby ensuring that entities with disjoint uncertainty regions do not encroach one another), and (ii) changing $X_d$ to be the largest $x$ such that $c_{d, x} \ge \Xi-2$ (since for $x$ exceeding this changed $X_d$ there can be at most $x-1$ entities that are within the encroachment threshold of any fixed entity).

\paragraph*{Relaxing the assumption of uniform entity extent}
We have assumed that all entities are $d$-dimensional balls with the same extent (radius). Relaxing this assumption impacts the relationship between the $x$-radius and $x$-separation of entities, captured in Equation~\eqref{eqn:lambda}. 
Nevertheless, if entity extents differ by at most a constant factor, it is straightforward to modify the constants $\lambda_{d, x}$ and $X_d$, so that all of our results continue to hold. 

\paragraph*{Relaxing the assumption of uniform entity speed}
Similarly, the reader will not be surprised by the fact that our results are essentially unchanged if our assumption that all entities have the same (unit) bound on their maximum speed is relaxed to allow speed bounds that differ by at most a constant factor. 
Allowing non-constant factor differences in speed bounds creates some additional challenges and some helpful new insights into the nature of our results. Surprisingly perhaps, our query schemes can be generalized in a rather straightforward way to accommodate such a change. 
However, our competitiveness results for maintaining bounded uncertainty degree only hold with respect to other (even clairvoyant) schemes with the same objective. 
This is unavoidable
since, when entities have arbitrarily different maximum speeds, 
uncertainty configurations with maximum ply $x$ could have arbitrarily large degree.

\paragraph*{Exploiting uniformity of entity distributions}
As previously mentioned, our competitive bounds can be further strengthened by making assumptions about the \emph{uniformity} of entity distributions.  
Define 
$\mu_{\E,x} = \frac{2^{1/d} \sum_{e_i \in \E}  \frac{1}{r_i(2 x)}}{\sum_{e_i \in \E}  \frac{1}{r_i(x)}}$.
The expression $\mu_{\E,x}$ can be viewed as a measure of the  $x$-uniformity of the distribution of entities in $\E$ at one fixed moment in time:
a completely uniform distribution would have $\mu_{\E,x} \approx 1$ and a large collection of isolated $x$-clusters could have $\mu_{\E,x}$ arbitrarily small. 
The uniformity-sensitivity of our query schemes is expressed in terms of $\mu_{\E,x}$.
More generally, define
\[
\mu_{\E,x}(T)= \frac{ 2^{1/d} \sum_{e_i \in \E} \int_T \frac{\dif{t}}{r_i(2x, t)}}
 { \sum_{e_i \in \E} \int_T \frac{\dif{t}}{r_i(x, t)}}
\]
as a measure of uniformity of the entity configurations over time interval $T$.

Revisiting the proof of Theorem~\ref{thm:FWRRplycompetitive}, we see that the ${\rm FWRR}[x]$ query scheme   uses a total number of queries over any 
time interval $T$
that is competitive
with any scheme that maintains uncertainty configurations with uncertainty ply at most $x$ over $T$, 
and the competitive factor is
$O(\frac{1}{\mu_{\E,x}})$, 
provided $|T| \phi_{\E, 2x} \ge 2|\E|$.
This follows because
\[
\phi_{\E, 2x} \ge \sum_{e_i \in \E}  \frac{1}{r_i(2 x)}
= \frac{\mu_{\E,x}}{2^{1/d}} \sum_{e_i \in \E}  \frac{1}{r_i(x)}
\ge \frac{\mu_{\E,x} \lambda_{d,x}}{2^{1/d}} \phi_{\E, x}.
\]

Similarly, by revisiting the proof of Theorem~\ref{thm:fulladaptive}, we see that,
for all sufficiently large time intervals $T$,
the refined BucketScheme$[x]$ 
uses a of total number of queries over $T$, or some small shift $\overrightarrow{T}$ of $T$, that is competitive with any query scheme that maintains uncertainty ply at most $x$ over $T \cup \overrightarrow{T}$.
The competitive factor is
$O(\frac{1}{\mu_{\E,x}(T)})$ when 
the entity set $\E$ has $x$-uniformity $\mu_{\E,x}(T)$ over $T$.
In this case we make use of the fact that 
$\phi_{\E, 2x}(T) \ge \frac{\mu_{\E,x}(T) \lambda_{d,x}}{2^{1/d}} \phi_{\E, x}(T)$.

\paragraph*{Decentralization of our query schemes}

Our query model implicitly assumes that query decisions are centralized; location queries are issued from a single source. 
However, it is not hard to see that our query schemes, assuming suitable initial synchronization, could be fully decentralized, with location queries to entity $e_i$ replaced by location broadcasts from entity $e_i$. 
This is particularly interesting in the situation where entity speed
bounds could 
differ, 
since an entity can decide when to make its next location broadcast without knowing the speed bounds associated with other entities.

\bibliography {refs}

\appendix

%%%%% Appendices %%%%%%%%%%%

\section{Proof of Lemma~\ref{lem:ballcoverNEW}}\label{app:BallCoverNEW}

\LemBallCoverNEW*

\begin{proof}
(The structure of our proof resembles that used in Lemma 1 of \cite{BEK2019}.)
Let 
$E \subseteq \E$ be the set of entities $e_i$ whose $\hat{x}$-ball $B^Z_i(\hat{x})$ 
intersects $e_*$.
Let $\Gamma^Z_i(\hat{x})$ denote
the set of at most $\hat{x}$ entities
(including $e_i$) in $\E$ that intersect (as balls) the interior of $B^Z_i(\hat{x})$.
We construct a subset $I \subseteq E$ 
incrementally
by (i) selecting the un-eliminated entity $e_i \in E$ with the largest $\hat{x}$-radius, (ii) eliminating all entities in $\Gamma^Z_i(\hat{x})$ (including $e_i$ itself) from $E$, and (iii) repeating, until $E$ is empty.
It follows that (i) since we eliminate in non-increasing order of $\hat{x}$-radius, for all $e_i \in I$,
the interior of $B^Z_i(\hat{x})$
intersects only one entity from $I$ (namely, $e_i$), (ii) since 
$\Gamma^Z_i(\hat{x})$ contains at most $\hat{x}$ entities from $E$, $|I| \geq |E|/\hat{x}$, and (iii) if $|I|>1$ then 
$e_* \not \in I$.

Suppose that $|I|>1$; otherwise $|E| \le \hat{x}$, and nothing remains to be proved. Let $e_i$ and $e_j$ be any two distinct elements of $I$.
Then $\|z _i - z_j\| 
\ge \max \{ r^Z_i(\hat{x}) + 1, \;r^Z_j(\hat{x}) + 1 \}
\ge \max \{\|z_i - z_*\|, \;\|z_j - z_*\| \}$,
where $z_*$ denotes the centre of $e_*$.
Thus, the vector $v_i$ from $z_*$ to $z_i$ and the vector $v_j$ from $z_*$ to $z_j$ must form an angle at least $\pi/3 = 2\arcsin(1/2)$ 
(which occurs when $\|z _i - z_j\| = 
\|z_i - z_*\| = \|z_j - z_*\|$). 
But there can be no more than $5^d$ vectors from $z_*$ in 
$\mathbb{R}^d$ whose pairwise separation is at least 
$2\arcsin(1/2)$. (To see this, note that (i) balls of radius $1$ at distance $4$ from $z_*$ along each vector, must be disjoint, and (ii) there are at most $5^d$ disjoint balls of unit radius inside a ball of radius $5$ centred at $z_*$.)
Hence $|E| \leq 5^d \hat{x}$.
\end{proof}

\section{Query Optimization at a Fixed Target Time}\label{app:FTT}

If entities are stationary, the intrinsic fixed-target granularity, for any congestion measure, can in principle be realized, up to a factor of two, by an algorithm that first queries all of the entities, until half of the time to the target has expired (using granularity 
$\frac{\tau}{2 |\E|}$)
and then optimizes, with respect to the given congestion measure, taking advantage of the knowledge of entity positions.

Despite the fact that realizing the intrinsic fixed-target granularity exactly is NP-hard (this follows directly from Theorem~2.2 of~\cite{EKLS-SoCG13}),
for stationary entities it is straightforward to achieve degree at most $x$ at a fixed target time using query granularity that is within a constant factor of the intrinsic fixed-target granularity. 
This follows from the observations that 
(i) every entity must be queried within time corresponding to its $x$-separation,
prior to the target time, in order to achieve degree at most $x$ at the target time, and 
(ii) if every entity is queried within time corresponding to a fraction $h$ %one-third 
of its $x$-separation, prior to the target time, then the uncertainty degree at the target time is guaranteed to be at most $x$, provided $h < \lambda_{d,x}/3$.
The latter follows because
any entity $e_j$ whose separation from $e_i$ is $s \ge \sigma_i(x)$ 
has the property that $\sigma_j(x) \le s + 2
+ \sigma_i(x) \le 2s + 2$,
by equation~(\ref{eqn:observationA}), and
hence, %(otherwise the $x$-neighbourhood of $e_j$ would contain the $x$-neighbourhood of $e_i$), 
using equation~(\ref{eqn:lambda}),
\[ 
\sigma_i(x) + \sigma_j(x) \le 3s + 2 
\le (3 + 2\frac{1- \lambda_{d,x}}{\lambda_{d,x}}) s = \frac{\lambda_{d,x} + 2}{\lambda_{d,x}} s < s/h.
\]

Guaranteeing that a particular measure of congestion potential is at most some specified value $x$ at the fixed target time $\tau$ 
%time units into the future 
is obviously more of a challenge for mobile entities.
We define the {\em projected uncertainty region} of an entity, at any moment in time, to be the uncertainty region for that entity that would result if the entity experiences no further queries before the target time.  
We say that an entity is {\em $x$-ply-safe} at a particular time if its projected uncertainty region 
has no point in common with
the projected uncertainty regions of more than $x-1$ other entities (so that it could not possibly have uncertainty ply exceeding $x$ at the target time).
Similarly, we say that an entity is {\em $x$-degree-safe} at a particular time if its projected uncertainty region 
intersects 
the projected uncertainty regions of at most $x-1$ other entities (so that its $x$-separation at the target time is guaranteed to be positive).

The following Fixed-Target-Time(FTT) query scheme 
shows that uncertainty degree at most $x\!+\!\Delta$ can be guaranteed at a fixed target time 
using minimum query granularity that is at most
$\Theta(\frac{x\!+\!\Delta}{1\!+\!\Delta})$ smaller than
that used by any query scheme that guarantees uncertainty degree at most $x$.
Since the uncertainty regions of all entities are unbounded at time $0$, 
none of the entities are $(x\!+\!\Delta)$-degree-safe to start (assuming $x\!+\!\Delta < n$). 
Furthermore any scheme, including a clairvoyant scheme, must 
query all but $x$ of the entities at least once in order to avoid ply greater than $x$ at the target time.
The FTT$[x\!+\!\Delta]$ scheme starts by querying all entities in a single round 
using query granularity $\tau/ (2n)$,
which is $O(1)$-competitive, assuming $n-(x\!+\!\Delta) = \Omega(x\!+\!\Delta)$, with what must be done by any other scheme. 

At this point, the FTT$[x\!+\!\Delta]$ scheme identifies two sets of entities (i) the $n_1$ entities that are not yet $(x\!+\!\Delta)$-degree-safe (the \emph{unsafe survivors}), and (ii) the $m_1$ entities that are $(x\!+\!\Delta)$-degree-safe and whose projected uncertainty region intersects the projected uncertainty region of one or more of the unsafe survivors (the \emph{safe survivors}). 
All other entities are set aside and attract no further queries.

The scheme then queries, in a second round, all $n_1 + m_1$ survivors
using query granularity $\frac{\tau}{4( n_1 + m_1)}$. 
In general, after the $r$th round, the scheme 
identifies $n_r$ unsafe survivors and $m_r$ safe survivors, which,
assuming $n_r + m_r >0$,
continue into an $(r+1)$st round using granularity 
$\frac{\tau}{2^{r+1}( n_r + m_r)}$.
The $r$th round completes at time $\tau\!-\!\tau/2^r$. 
Furthermore, all entities that have not been set aside have a projected uncertainty region whose radius is in the range
$(1 + \tau /2^r, 1 + \tau/2^{r-1}]$.

\FixedTargetTimeTheorem*

\begin{proof}
We claim that any query
scheme 
$\AAA$
that guarantees uncertainty degree at most $x$ at time $\tau$ must use at least 
$\Theta(\frac{(n_r + m_r)(1\!+\!\Delta)}{x\!+\!\Delta})$ 
queries after the start of the $r$th query round; any fewer queries
would 
result in one or more entities having degree greater than $x$ at the target time.

To see this observe first that each of the $n_r$ unsafe survivors is either queried by $\AAA$ after the start of the $r$th query round or has its projected uncertainty degree reduced below $x$ 
by at least $1\!+\!\Delta$ 
queries to
its projected uncertainty neighbours after the start of the $r$th query round. 
Assuming that fewer than $n_r/2$ unsafe survivors are queried by $\AAA$ after the start of the $r$th query round, we argue that 
at least $\frac{n_r (1\!+\!\Delta)}{2 \cdot 4^d (x\!+\!\Delta)}$ 
queries must be made after the start of the $r$th query round to reduce below $x$ the 
projected uncertainty degree of the remaining unsafe survivors.

Note that any query 
after the start of the $r$th round 
to an entity set aside in an earlier round 
cannot serve to 
lower the projected uncertainty degree of
any of the $n_r$ unsafe survivors.
Furthermore, any query to one of the survivors of the $(r-1)$st round can serve to
decrease by one the projected uncertainty degree of
at most $4^d (x\!+\!\Delta)$ of the unsafe survivors whose uncertainty degree is at most $x\!+\!\Delta$.
(This follows because (i) the projected uncertainty regions of all survivors are within a factor of $2$ in size, and (ii) any collection of $4^d \hat{x}$ unit radius balls that are all contained in a ball of radius $4$, must have ply at least $\hat{x}$.)
Thus any scheme
that guarantees uncertainty degree at most $x$ at time $\tau$ must make 
at least $\frac{n_r (1\!+\!\Delta)}{2 \cdot 4^d (x\!+\!\Delta)}$ 
queries after the start of the $r$th query round.

Similarly, observe that each of the $m_r$ safe survivors must have each of its unsafe neighbours satisfied in the sense described above. 
But,
since the projected uncertainty regions of all survivors are within a factor of $2$ in size, each query that serves to lower the projected uncertainty degree of an unsafe neighbour of some safe survivor $e_i$ must be to an entity $e_j$ that has the projected uncertainty region of $e_i$ in its projected uncertainty \emph{near-neighbourhood} (the ball centred at $z_j$, whose radius is $9$ times the projected uncertainty radius of $e_j$).
But $e_j$ has at most $18^d (x\!+\!\Delta)$ such safe near-neighbours, since
any collection of $18^d \hat{x}$ unit radius balls that are all contained in  a ball of radius $18$, must have ply at least $\hat{x}$.

It follows that, even if a query to $e_j$ lowers the projected uncertainty degree of all of the unsafe neighbours of $e_i$, a total of at least $\frac{m_r (1\!+\!\Delta)}{18^d (x\!+\!\Delta)}$ queries must be made 
after the start of the $r$th query round
by any scheme
that guarantees uncertainty degree at most $x$ at time $\tau$.

Thus, 
query scheme $\AAA$
must use at least 
$\max\{\frac{n_r (1\!+\!\Delta)}{2 \cdot 4^d (x\!+\!\Delta)}),\;
\frac{m_r (1\!+\!\Delta)}{18^d (x\!+\!\Delta)} \}\\
= \max\{\frac{n_r}{2 \cdot 4^d}, \; \frac{m_r}{18^d}\}
\frac{1\!+\!\Delta}{x\!+\!\Delta}
\ge \frac{n_r + m_r}{2 \cdot (18)^d} \frac{1\!+\!\Delta}{x\!+\!\Delta}$
queries over the interval 
$[\tau - \tau/2^{r-1},  \tau]$.
It follows that our query scheme is 
$\Theta(\frac{x\!+\!\Delta }{ 1\!+\!\Delta})$-competitive, in terms of maximum query granularity, with any, even clairvoyant, query scheme that guarantees  uncertainty degree at most $x$ at the target time. 
\end{proof}

\section{
Limitations on Competitiveness: Stationary Entities
}\label{ex:stationaryPlyFreqX}

{\bf Fixed-time vs continuous-time optimization.}
Consider
a collection of $n/2$ well-separated pairs, where the $i$-th pair has separation $4i-1$. 
Uncertainty ply (and degree) can be kept at one at a deadline $n$ time units from the start, 
by a 
scheme that queries entities with granularity one in decreasing order of their separation.
On the other hand 
to maintain degree/ply one continuously, 
the $i$-th entity pair must be queried at least once every  $4i-1$ steps, 
so over any time interval of length $n$, $\Omega(n \ln n)$ queries are required.

{\bf Maintaining ply vs degree.}
The example involves two clusters $A$ and $B$ of $(x +\Delta +1)/2$ point entities separated by distance $4(1\!+\!\Delta)$.  To maintain uncertainty ply at most $x$ it suffices to query $\Delta +1$ entities in both clusters once every $2(1\!+\!\Delta)$ steps, which can be achieved with query frequency one.
Since the uncertainty regions associated with queried points in cluster $A$ never intersect the  uncertainty regions associated with queried points in cluster $B$, the largest ply possible involves points in one cluster (say $A$) together with unqueried points in the other cluster ($B$), for a total of $x$. 

On the other hand, to maintain degree at most $x\!+\!\Delta$ no uncertainty region can be allowed to have radius $4(1\!+\!\Delta)$. Thus all $x\!+\!\Delta +1$ entities need to be queried with frequency at least $1/(4(1\!+\!\Delta))$, giving a total query demand of $x\!+\!\Delta +1$ over any time interval of length $4(1\!+\!\Delta)$. 

\section{
Limitations on Competitiveness: Moving Entities
}\label{ex:MobilePlyFreqX}

{\bf The fixed target time case.}
Imagine a configuration involving two collections $A$ and $B$ each with  $(x +1)/2 + \Delta$ point entities located in $\R^1$, on opposite sides of  a point $O$. At time $0$ all of the entities are at distance $x + 4 \Delta + 4$ from $O$, but have unbounded uncertainty regions. 
All entities begin by moving towards $O$ at unit speed, but at time $2 \Delta + 3$ a subset of $\Delta +1$ entities in each of $A$ and $B$
(the {\em special} entities)  change direction and move away from $O$ at unit speed, while the others carry on until the target time 
$x + 4 \Delta + 4$ when they simultaneously reach $O$ and stop.

To avoid uncertainty degree greater than $(x\!+\!\Delta)$ at a target time a clairvoyant algorithm needs only to (i) query all entities (in arbitrary order) up to time $x +2 \Delta + 1 $, and then (ii) query just the special entities (in arbitrary order)  in the next $2 \Delta + 2$ time prior to the target, using query granularity $1$, since doing so will leave the uncertainty regions of the points in $A$ disjoint from the special points in $B$, and vice versa.

On the other hand, to avoid ply $x$ at the target time any algorithm must query all special entities in both $A$ and $B$ in the last 
$2 \Delta + 3$ time before the target,
which in the absence of knowledge about which entities are special requires all entities in both $A$ and $B$ to be queried in the worst case, requiring query granularity at most $\frac{2 \Delta + 3}{x +1 +2\Delta}$.

Thus every algorithm that achieves congestion 
ply at most $x$ 
at the target time needs to use at least a factor $\frac{x+2\Delta +1}{ 2\Delta + 3}$
smaller query granularity on some instances than the best query scheme for achieving congestion 
degree at most $(x\!+\!\Delta) $
at the target time on those same instances.

{\bf The continuous case.}
A very similar construction allows us to conclude essentially the same result in the continuous case: every algorithm that maintains congestion 
ply at most $x$
at all times needs to use at least a factor $\frac{x+2\Delta +1}{2\Delta + 3}$ smaller query granularity on some instances over some time interval than the best query scheme for achieving congestion 
degree at most $x\!+\!\Delta $
at all times on those same instances.

It remains the case that to avoid ply $x$ just at the target time, any algorithm must query all special entities in both $A$ and $B$ before the target time, requiring query granularity at most $\frac{2 \Delta + 3}{x +1 +2\Delta}$.
Furthermore, 
even if there are no queries prior to what we called the target time, the fixed target time construction cannot lead to uncertainty degree more than $(x +1)/2 + \Delta$ prior to that time. 
So it remains to show that a clairvoyant algorithm can continue to maintain degree at most $x\!+\!\Delta$ thereafter using query granularity 
$1$ (i.e. the query scheme has not simply deferred a situation in which 
smaller query granularity is needed to maintain degree $x\!+\!\Delta)$. 
To see this, note that every time one of the special entities is queried it is seen to have distance exactly $2 \Delta + 3$ from the uncertainty regions associated with non-special entities on the other side of $O$, up to the time the non-special entities reach $O$, and at least $2 \Delta + 3$ thereafter, even if non-special entities are never queried. 
Thus it certainly suffices to query all $2\Delta + 2$ of the special entities at least once every $2\Delta + 2$ time steps to guarantee that their associated uncertainty regions remain disjoint from the uncertainty regions of all non-special entities on the other side of $O$.
This can %clearly
be done using query granularity $1$.

There might be some concern that (i) the construction assumes $n$, the total number of entities, is 
$x +1 + 2 \Delta$, and (ii) the competitive gap demonstrated by this example is only transitory, since after the target time a non-clairvoyant scheme could also maintain low congestion using granularity $1$. 
However, it is straightforward to modify the construction, by well-separated replication, to make $n$ arbitrarily large relative to $x$. 
Furthermore, by having the non-special entities retreat from $O$ at unit speed after the target time, one can essentially recreate the initial configuration and thereafter reproduce the high competitive gap periodically.

\section{Proof of Lemma~\ref{lem:mobilePlyLB}}
\label{sec:MobilePlyLBLemma}

Before detailing the proof of Lemma~\ref{lem:mobilePlyLB} we present another geometric lemma, demonstrating that Lemma~\ref{lem:ballcoverNEW} can be generalized to apply to $\hat{x}$-balls that have been scaled by some constant factor $\alpha \ge 1$, referred to as \emph{$\alpha$-inflated} $\hat{x}$-balls, provided these $\hat{x}$-balls are all comparable (to within a constant factor) in size.

\begin{lemma}\label{lem:inflballcover}
Let $\alpha,\gamma \ge 1$.
In any $\E$-configuration $Z$, any ball $B$ with radius $r_B$ intersects the $\alpha$-inflated $\hat{x}$-balls of at most 
$(2 \alpha \gamma +3)^d \hat{x}$ entities whose $\hat{x}$-balls have a radius in the range $[r_B^+, \gamma r_B^+]$, where $r_B^+ \ge r_B$. 
\end{lemma}

\begin{proof}
Let $B$ be a ball with centre $z_B$ and radius $r_B$, and let 
$E \subseteq \E$ be the set of entities $e_i$ whose $\alpha$-inflated $\hat{x}$-ball $\alpha B^Z_i(\hat{x})$ 
intersects $B$ and whose $\hat{x}$-radius $r^Z_i(\hat{x})$ 
lies in the range $[r_B^+, \gamma r_B^+]$, where $r_B^+ \ge r_B$.
Let $\Gamma^Z_i(\hat{x})$ denote
the set of at most $\hat{x}$ entities
(including $e_i$) in $\E$ that intersect (as balls) the interior of $B^Z_i(\hat{x})$.
We construct a subset $I \subseteq E$ 
incrementally
by (i) selecting the un-eliminated entity $e_i \in E$ with the largest $\hat{x}$-radius, (ii) eliminating all entities in $\Gamma^Z_i(\hat{x})$ (including $e_i$ itself) from $E$, and (iii) repeating, until $E$ is empty.
It follows that (i) since we eliminate entities in non-increasing order of $\hat{x}$-radius, for all $e_i \in I$,
the interior of $B^Z_i(\hat{x})$
intersects only one entity from $I$ (namely, $e_i$), and (ii) since 
$\Gamma^Z_i(\hat{x})$ contains at most $\hat{x}$ entities from $E$, $|I| \geq |E|/\hat{x}$.

For each $e_i \in I$, let $\hat{B}_i$ denote the ball centred at $z_i$ with radius $r_B^+/2$. 
Since the balls $\hat{B}_i$ are all disjoint and all lie entirely within a ball of radius $r_B +  \alpha \gamma  r_B^+ + r_B^+/2 \le (2 \alpha \gamma  + 3)r_B^+/2 $, centred at the centre of $B$, it follows that 
$|I| \le (2 \alpha \gamma +3)^d$.
\end{proof}

\MobilePlyLBLemma*

\begin{proof}
At a high level the proof parallels that of Lemma~\ref{newnewlem:stationaryPlyLB}.
For each  $e_i \in \E$ partition $T$ into sub-intervals, starting at times $t_1 < t_2 < \cdots$, whose length depends on (i) the $(x\!+\!\Delta)$-radius of $e_i$ at the start of the sub-interval,
and (ii) the uniformity of the $(x\!+\!\Delta)$-radius of $e_i$ throughout the interval. 
Since the separation between any two entities changes by at most two for each unit of time, the $(x\!+\!\Delta)$-radius of $e_i$ changes over one of its associated sub-intervals by at most two times the length of that sub-interval.
Sub-intervals that end at the end of $T$ are referred to as {\em terminal} sub-intervals; 
other sub-intervals have one of two types.
{\em Good} sub-intervals $[t_j, t_{j+1})$ have length 
$t_{j+1}-t_j = r_i(x\!+\!\Delta, t_j)$ and maintain an $x$-radius greater than  
$r_i(x\!+\!\Delta, t_j)/3$ throughout. 
{\em Bad} sub-intervals $[t_j, t_{j+1})$ have length 
$t_{j+1}-t_j \leq r_i(x\!+\!\Delta, t_j)$ 
and end (after at least 
$r_i(x\!+\!\Delta, t_j)/3$ time units) the first time the 
$x$-radius is less than or equal to 
$r_i(x\!+\!\Delta, t_j)/3$. 
We refer to the interval between 
$r_i(x\!+\!\Delta, t_j)$ and 
$r_i(x\!+\!\Delta, t_{j+1})$ as the \emph{span}
of sub-interval $[t_j, t_{j+1})$.
(See Fig.~\ref{fig:goodbadinterval}.)

Thus, good sub-intervals satisfy
$r_i(x\!+\!\Delta, t_j)/3 < r_i(x\!+\!\Delta, t) \le 3 r_i(x\!+\!\Delta, t_j)$, for $t_j \le t < t_{j+1}$,
and bad sub-intervals satisfy
$r_i(x\!+\!\Delta, t_j)/3 < r_i(x\!+\!\Delta, t) < 3 r_i(x\!+\!\Delta, t_j)$, for $t_j \le t < t_{j+1}$, and 
$r_i(x\!+\!\Delta, t_{j+1}) = r_i(x\!+\!\Delta, t_j)/3$.
Hence, 
for good sub-intervals 
$1/3 \leq \int_{t_j}^{t_{j+1}} \frac{\dif t}{r_i(x\!+\!\Delta, t)} < 3$, and for bad sub-intervals 
$1/9 < \int_{t_j}^{t_{j+1}} \frac{\dif t}{r_i(x\!+\!\Delta, t)} < 3$.
Note that in any consecutive sequence of non-terminal sub-intervals, a good sub-interval of length
$r_i(x\!+\!\Delta, t)$ is followed by a consecutive sequence of bad sub-intervals of total length at most 
$\sum_{k \ge 0} [3 r_i(x\!+\!\Delta, t)/3^k] 
\le 9 r_i(x\!+\!\Delta, t)/ 2$.

\begin{figure}
\centering
\includegraphics{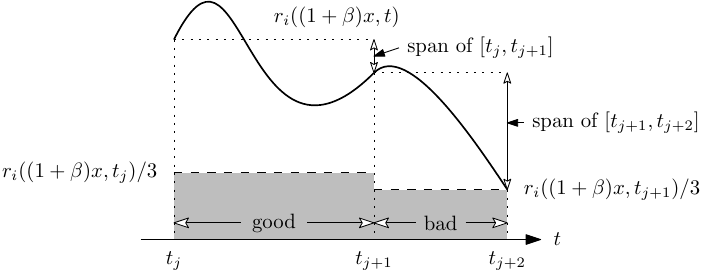}
\caption{An example of a good and bad sub-interval of the function $r_i(x\!+\!\Delta,t)$.}
\label{fig:goodbadinterval}
\end{figure}

\begin{sloppypar}
If we ignore all terminal sub-intervals and all entities with fewer than $100$ non-terminal sub-intervals, the sub-intervals of the remaining entities contribute in total 
more than $2/5 \sum_{e_i \in \E} \int_T  \frac{\dif t}{r_i(x\!+\!\Delta, t)}$
to $\phi_{\E,x}(T)$. 
It follows that there must be a total of at least
$2/15 \sum_{e_i \in \E} \int_T  \frac{\dif t}{r_i(x\!+\!\Delta, t)}$
such intervals.
\end{sloppypar}

Let $e_i$ be one of the remaining entities, with $w_i \ge 100$ non-terminal sub-intervals.
A bad sub-interval $[t_j, t_{j+1})$ is {\em matched} if there is a subsequent good sub-interval whose 
span
includes  $r_i(x\!+\!\Delta, t_{j+1})$. 
Since the 
span length  of
a given good sub-interval $[t_j, t_{j+1})$ is at most 
$2r_i(x\!+\!\Delta, t_j)$
and the 
span length of
a given bad sub-interval $[t_j, t_{j+1})$ is at least 
$2r_i(x\!+\!\Delta, t_{j+1})$
it follows that a good sub-interval can serve as the earliest match to at most one bad sub-interval. 

Suppose that more than $4w_i/5$ of the non-terminal sub-intervals of $e_i$ over $T$ are unmatched bad intervals. 
Then $e_i$'s sub-interval sequence 
must end with a sub-interval of length at most $|T|/ 3^{4w_i/5}$. 
In this case consider the continuation of $e_i$'s sub-interval sequence into the remainder of $T^+$.
If this includes an interval of length $|T|/ 3^{3w_i/5}$ then this extension must include at least $w_i/5$ good intervals. 
But if all sub-intervals in the continuation are smaller than $|T|/ 3^{3w_i/5} $ then 
the unmatched bad intervals among these span less than $2 |T|/ 3^{3w_i/5}$,
so there must be at least $\frac{|T^+|-|T| - 2 |T|/ 3^{3w_i/5}}{|T|/ 3^{3w_i/5}} > w_i -2 > w_i/2$ 
good and matched bad intervals over the extension
(since $|T|/ 3^{55} > w_i \cdot |T|/  3^{3w_i/5}$, when $w_i \ge 100$). 

Alternatively, at least $w_i/5$ of the non-terminal sub-intervals of $e_i$ over $T$ are good or matched bad intervals.
In either case, it must be that at least $w_i/5$ of the sub-intervals of $e_i$ over $T^+$ are either good or matched bad intervals. 
Hence at least $w_i/10$ of these sub-intervals  must be good.  
Since this is true of all of the entities with at least $100$ non-terminal sub-intervals, and such entities have at least $2/15 \sum_{e_i \in \E} \int_T  \frac{\dif t}{r_i(x\!+\!\Delta, t)}$
non-terminal sub-intervals in total, it follows that there must be at least $2/150 \sum_{e_i \in \E} \int_T  \frac{\dif t}{r_i(x\!+\!\Delta, t)}
$ good sub-intervals in total.

It remains then to follow the argument from the stationary case (Lemma~\ref{newnewlem:stationaryPlyLB}), focusing on good sub-intervals alone, since (a) these are long enough to support the argument that maintaining ply at most $x$ requires that all but $x-1$ initial $(x\!+\!\Delta)$-neighbours must be queried within the sub-interval in order to avoid ply greater than $x$,
and (b) the $(x\!+\!\Delta)$-radius is sufficiently uniform over these sub-intervals to support the assertion that any  query can contribute, 
\emph{on average}, 
to the demand of at most $\Theta(x\!+\!\Delta)$ sub-intervals.

If fewer than $\Delta + 1$ of the at least $x\!+\!\Delta +1$ entities (including $e_i$) that intersect the $(x\!+\!\Delta)$-ball of $e_i$ at the start of a good sub-interval
are not queried in that sub-interval
then at the end of that sub-interval at least $x+1$ of these 
entities must have uncertainty regions of radius at least $r_i(x\!+\!\Delta)$, and hence intersect at the point occupied by the centre of $e_i$ at the start of the sub-interval, forming ply at least $x+1$ at that point. 
Thus, to avoid ply greater than $x$ throughout $T$, every entity $e_i$ must 
have associated with it at least $\Delta + 1$ queries in each of its good sub-intervals. 
Summing over all good sub-intervals of all entities this gives a total of at least $\frac{\Delta +1}{75} \cdot  \sum_{e_i \in \E} \int_T  \frac{\dif t}{r_i(x\!+\!\Delta, t)}
\ge \frac{\Delta +1}{75} \cdot \lambda_{d,x\!+\!\Delta}\sum_{e_i \in \E} \int_T  \frac{\dif t}{\sigma_i(x\!+\!\Delta, t)}
= \frac{\Delta +1}{75} \cdot \lambda_{d,x\!+\!\Delta} \phi_{\E,x\!+\!\Delta}(T) $ associated queries. 

Of course, as in Lemma~\ref{newnewlem:stationaryPlyLB}, a query to some entity $e_j$ may serve to help satisfy sub-interval query demands for many different entities $e_i$.
In fact, unlike the situation with stationary entities, it is possible for one query to help satisfy good sub-interval query demands associated with an \emph{arbitrarily large} number of different entities.
The query sharing argument from Lemma~\ref{newnewlem:stationaryPlyLB} is
complicated in the dynamic setting by the fact that, at the time a query to entity $e_q$ helps satisfy a good sub-interval $G$ of entity $e_i$
it no longer necessarily intersects the $(x\!+\!\Delta)$-ball of $e_i$. 
Nevertheless, since the $(x\!+\!\Delta)$-radius of $e_i$ cannot shrink by more than a factor of three over $G$, and the separation of $e_i$ and $e_q$ cannot increase by more than two times the $(x\!+\!\Delta)$-radius of $e_i$ at the start of $G$, entity $e_q$ must intersect the $9$-inflated $(x\!+\!\Delta)$-ball of all of the entities $e_i$ that a query to $e_q$ helps satisfy.
This will allow us to conclude (using Lemma~\ref{lem:inflballcover}
with $\hat{x}= x\!+\!\Delta$) that there can only be $\Theta(x\!+\!\Delta)$ such satisfied sub-intervals that are comparable in length. 
Furthermore, we show that the query demand of any sub-interval that can be partially satisfied with a query that simultaneously helps satisfy the demand of more than some sufficiently large multiple of $x\!+\!\Delta$ smaller intervals, can be ``charged'' to smaller sub-intervals in such a way that no sub-interval accumulates a charge exceeding its initial demand.
In this way, it follows that the total number of distinct queries needed to satisfy the total query demand
$\frac{\Delta +1}{75} \cdot \lambda_{d,x\!+\!\Delta} \phi_{\E,x\!+\!\Delta}(T)$
associated with all good sub-intervals is at least a fraction $\frac{1}{\Theta(x\!+\!\Delta)}$ of that total demand.

Consider a good sub-interval $G$ of entity $e_i$ starting at time $t_G$.  
Entity $e_j$ is said to be \emph{close} to entity $e_i$ at time $t_G$ if their separation is at most $10 r_i(x\!+\!\Delta, t_G)/3$.
Note that if $e_j$ is close to $e_i$ at time $t_G$,
then over the entire interval 
$G^+ = [t_G, t_G + 10 r_i(x\!+\!\Delta, t_G)/9]$
(i.e over $G$ extended by $|G|/9$), 
it remains within distance $50 r_i(x\!+\!\Delta, t_G)/9$
 of $e_i$,
and hence within distance 
$50 r_i(x\!+\!\Delta, t_G)/9 + 1
\le 59 r_i(x\!+\!\Delta, t_G)/9$ of $z_i$
Furthermore, since $r_i(x\!+\!\Delta, t)$ shrinks by at most a factor of $9$ over $G^+$, it follows 
that $e_j$ continues to intersect the 
$59$-inflated 
$(x\!+\!\Delta)$-ball of $e_i$ throughout 
$G^+$.
In addition, if some good interval of $e_j$ of length at most $r_i(x\!+\!\Delta, t_G)/3$ intersects $G$ then good intervals of $e_j$ of length at most $r_i(x\!+\!\Delta, t_G)/3$ 
that intersect
the interval $G^+$, have total length at least 
$2 r_i(x\!+\!\Delta, t_G)/99$.
(This follows immediately from our earlier observation that the length of a good sub-interval is at least a fraction $2/9$ of the length of a subsequent sequence of bad sub-intervals, so at least a fraction $2/11$ of the total length of any consecutive sequence of sub-intervals starting with a good sub-interval is covered by good sub-intervals.)

We say that a good sub-interval $G$ of entity is \emph{heavy} if there are at least 
$99 (354)^d (x\!+\!\Delta)$
entities that are close to $e_i$ at time $t_G$, each of which has  at least one good sub-interval of length at most $r_i(x\!+\!\Delta, t_G)/3$ that intersects $G$. All other good sub-intervals are \emph{light}.
We re-allocate the charge associated with any heavy sub-interval $G$ (initially its $\Delta +1$ demand) to the good sub-intervals of the at least 
$99 (354)^d (x\!+\!\Delta)$ 
entities of length at most $r_i(x\!+\!\Delta, t_G)/3$ that intersect the interval $G^+$, in proportion to the length of each such sub-interval. 
Since the sub-intervals receiving a charge have total length at least 
$2 (354)^d (x\!+\!\Delta) r_i(x\!+\!\Delta, t_G)$, 
a sub-interval of length $\alpha r_i(x\!+\!\Delta, t_G)$ receives a fraction of at most 
$\frac{\alpha}{2 (354)^d (x\!+\!\Delta)}$
of the charge associated with $G$.

If sub-interval $H$ of entity $e_j$ receives a charge allocation from sub-interval $G$ of entity $e_i$ then at either the start or end (or both) of the sub-interval $H$, $e_j$ must intersect the $59$-inflated $(x\!+\!\Delta)$-ball of $e_i$. 
Thus, by Lemma~\ref{lem:inflballcover}
(choosing $\hat{x}= x\!+\!\Delta$, $\alpha=59$, and $\gamma=3$),  sub-interval $H$ receives a charge from at most $2 (354)^d (x\!+\!\Delta)$ different sub-intervals whose length is in the interval $[3^s |H|, 3^{s+1} |H|)$.
So the total charge (in the first phase of charge re-allocation) re-allocated to sub-interval $H$ is at most
$\sum_{s \ge 1} 2 (354)^d (x\!+\!\Delta) \frac{1/3^s}{2 (354)^d (x\!+\!\Delta)} (x\!+\!\Delta)
\le  (x\!+\!\Delta) \sum_{s \ge 1} \frac{1}{3^s} 
\le (x\!+\!\Delta) /2 $.
So if this charge re-allocation from heavy sub-intervals is repeated, then after the $a$-th re-allocation (i) the charge associated with each light sub-interval is at most $\frac{2^{a+1}-1}{2^a}$ times its initial charge, and
(ii) the change associated with each heavy sub-interval
is at most $\frac{1}{2^a}$ times its initial charge.
While there remain positively charged heavy sub-intervals each phase of charge reallocation must reduce to zero the charge of at least one heavy sub-interval, so after the charge associated with all heavy sub-intervals has been reduced to zero, 
(i) all remaining (light) sub-intervals have charge less than $2(\Delta +1)$, and (ii) their total charge equals the total initial demand associated with all good sub-intervals. 
It follows that the total query demand to satisfy all light sub-intervals is at least half that required to satisfy all good sub-intervals, i.e. at least
$\frac{\Delta +1}{150} \cdot \lambda_{d,x\!+\!\Delta} \phi_{\E,x\!+\!\Delta}(T)$.

Any collection of light sub-intervals whose query demands are partially satisfied by the same query has size  at most $[(99 (354)^d + (489)^d] (x\!+\!\Delta)$ . 
To see this, suppose that 
$G$, a sub-interval of $e_i$, is the longest sub-interval in some collection of more than 
$[(99 (354)^d + (489)^d] (x\!+\!\Delta)$ 
light sub-intervals whose query demands are partially satisfied by the same query to entity $e_q$,
and $H$, a sub-interval of $e_j$, is another sub-interval in this collection.
Note that at the time of the query to $e_q$, 
(i) $e_q$ intersects the $9$-inflated $x$-ball of both $e_i$ and $e_j$, and 
(ii) the $x$-radii of $e_i$ and $e_j$ differ by at most a factor of $27$.
Thus, by Lemma~\ref{lem:inflballcover}
(choosing $\hat{x}= x\!+\!\Delta$, $\alpha=9$ and $\gamma=27$), at most  $(489)^d (x\!+\!\Delta)$
sub-intervals $H$ of length at least $|G|/3$ can be partially satisfied by the same query. 
So more than $99 (354)^d (x\!+\!\Delta)$ sub-intervals of length less than $|G|/3$ are partially satisfied by the same query. 
But each such sub-interval $H$ must intersect $G$. 
Furthermore, $H$ must be associated with an entity that is close to $e_i$ at the start of $G$. 
(To see this, observe that if a query to some entity $e_q$ helps satisfy both $G$ and $H$, then $e_q$ has distance at most $r_i(x\!+\!\Delta, t_G)$ from $z_i$ at time $t_G$ (the start of $G$) 
and distance at most $r_j(x\!+\!\Delta, t_H)$ from $z_j$ at time $t_H$ (the start of $H$). Hence 
$e_q$ has distance at most 
$r_j(x\!+\!\Delta, t_H) + 2|t_G - t_H| 
\le r_i(x\!+\!\Delta, t_G)/3 + 2 r_i(x\!+\!\Delta, t_G)$ from $z_j$ at time $t_G$, which implies that $e_i$ and $e_j$ are separated by at most 
$10 r_i(x\!+\!\Delta, t_G)/3$ at time $t_G$.)
Thus sub-interval $G$ must be heavy, contradicting our assumption.

It follows then that the number of distinct queries needed to fully satisfy the query demands of all light sub-intervals is at least
$\frac{1}{[(99 (354)^d + (489)^d] (x\!+\!\Delta)}
\frac{\Delta +1}{75} \cdot \lambda_{d,x\!+\!\Delta} \phi_{\E,x\!+\!\Delta}(T)$.
\end{proof}

\section{Proof of Corollary~\ref{cor:mobilePlyLB}}
\label{sec:ProofMobilePlyLBCor}

\MobilePlyLBCor*

\begin{proof}
Let $\overleftarrow{T}$ be the interval $T$ shifted by $-|T|/3^{55}$,
$T_0 = \overleftarrow{T} \cap T$, and $T_1 = \overrightarrow{T} \cap T$.
Since  $T = T_0 \cup T_1$, 
either 
 $\sum_{e_i \in \E} \int_{T_0}  \frac{\dif t}{r_i(x\!+\!\Delta, t)} 
 \ge  1/2 \sum_{e_i \in \E} \int_{T_0}  \frac{\dif t}{r_i(x\!+\!\Delta, t)} 
 \ge 505 |\E|$
or  $\sum_{e_i \in \E} \int_{T_1}  \frac{\dif t}{r_i(x\!+\!\Delta, t)} 
\ge  1/2 \sum_{e_i \in \E} \int_{T_0}  \frac{\dif t}{r_i(x\!+\!\Delta, t)} 
\ge 505 |\E|$.
But $T_0^+ \subset T$ and $T_1^+ \subset \overrightarrow{T}$, so 
any query scheme with maximum uncertainty ply at most $x$ over $T^+$ has
maximum uncertainty ply at most $x$ over both $T_0^+$ and $T_1^+$.
Hence, by the lemma, any 
query scheme with maximum uncertainty ply at most $x$ over $T^+$ must make 
a total of at least 
$\frac{(\Delta +1) \lambda_{d,x\!+\!\Delta}}{75 [(99 (354)^d + (489)^d] (x\!+\!\Delta)}
\phi_{\E,x\!+\!\Delta}(T)$
queries over either $T$ or $\overrightarrow{T}$.
\end{proof}

\section{Proof of Lemma~\ref{lem:NEWPerception}}
\label{sec:NEWPerceptionLemma}

\NEWNEWLemmaPerceptionPercQuery*

\begin{proof}
We begin by arguing by induction that for all entities $e_i$,
$\sd[t']{x}{i}/2 \leq \psd[t']{x}{i} \leq 3\sd[t']{x}{i}/2$
at all times $t' \ge t_0$ that entity $e_i$ is queried. 
If this is not true, suppose that $t^*$ is the time of the first query (without loss of generality, to $e_j$) after $t_0$ at which 
$\sd[t^*]{x}{j}/2 \leq \psd[t^*]{x}{j} \leq 3\sd[t^*]{x}{j}/2$ 
does not hold. 

Since 
$\sd[p_i(t^*)]{x}{i}/2 \leq \psd[p_i(t^*)]{x}{i} \leq 3\sd[p_i(t^*)]{x}{i}/2$, for all entities $e_i$, and
$\sd[p_i(t)]{x}{i} - 2(t - p_i(t)) \leq \sd[t]{x}{i} \leq \sd[p_i(t)]{x}{i} + 2(t - p_i(t))$,
for $t > p_i(t)$, 
it follows from assumption (ii) that
\[
t^* - p_i(t^*) \leq
\frac{ \lambda_{d,x} \psd[p_i(t^*)]{x}{i}}{12} \leq 
\frac{\lambda_{d,x} \sd[p_i(t^*)]{x}{i}}{8} \leq 
\frac{\lambda_{d,x}(\sd[t^*]{x}{i} + 2(t^*-p_i(t^*)))}{8}
\]
which implies
$\frac{4 - \lambda_{d,x}}{4} (t^* - p_i(t^*)) \leq \frac{\lambda_{d,x} \sd[t^*]{x}{i}}{8}$ and hence 
$t^* - p_i(t^*) 
\leq \frac{\lambda_{d,x}}{8- 2\lambda_{d,x}} \sd[t^*]{x}{i}
\leq \frac{\lambda_{d,x}}{6} \sd[t^*]{x}{i}$ .
Using  
Equation~(\ref{eqn:observationA}),
it follows that
\[
t^*-p_i(t^*) 
\leq 
\frac{\lambda_{d,x}(\|\zeta_i(t^*) - \zeta_j(t^*)\| + \sd[t^*]{x}{j})}{6} .
\]
So at time $t^*$ every entity $e_i$, where $i \neq j$, has been queried within the last 
$\lambda_{d,x} (\|\zeta_i(t^*) - \zeta_j(t^*)\| + \sd[t^*]{x}{i})/6$ time steps and has a
perceived position within this distance of its actual position at time $t^*$.
Since the perceived position of $e_j$ may also differ, by at most
$\lambda_{d,x} \sd[t^*]{x}{j}/6$, from its actual position at time $t^*$, the distance
between the perceived locations of $e_i$ and $e_j$ at time $t^*$ satisfies,

\begin{align*}
 (1 -\lambda_{d,x}/6 )&\|\zeta_i(t^*) - \zeta_j(t^*)\| - 2 \lambda_{d,x} \sd[t^*]{x}{j})/6\\
 &\le
\|\zeta_i(p_i(t^*)) - \zeta_j(p_j(t^*))\| \\
&\le 
 (1 +\lambda_{d,x}/6 )\|\zeta_i(t^*) - \zeta_j(t^*)\| + 2 \lambda_{d,x} \sd[t^*]{x}{j})/6
\end{align*}

Since at least $x$ entities $e_i$ (the ones, other than $e_j$,
in $\NN[t^*]{x}{j}$) satisfy 
$\|\zeta_i(t^*) - \zeta_j(t^*)\| \le \sd[t^*]{x}{j} + 2 $,
these same entities 
have perceived locations at time $t^*$ within
distance 
$(1 +\lambda_{d,x}/6 )(\sd[t^*]{x}{j} + 2 ) + 2 \lambda_{d,x} \sd[t^*]{x}{j})/6 
= (1 +3\lambda_{d,x}/6 ) \sd[t^*]{x}{j} + 
2(1 +\lambda_{d,x}/6 ) $
of $\zeta_j(p_j(t^*))$. 
Thus 

\begin{align*}
\psd[t^*]{x}{j} 
&\le (1 +3\lambda_{d,x}/6 ) \sd[t^*]{x}{j} + \lambda_{d,x}  /3 \\
&\le (1 +3\lambda_{d,x}/6 ) \sd[t^*]{x}{j} + (1 - \lambda_{d,x}) \sd[t^*]{x}{j} /3 \\
&\le (4/3 + \lambda_{d,x}/6) \sd[t^*]{x}{j} 
\le 3\sd[t^*]{x}{j} / 2.
\end{align*}

On the other hand, 
since all but at most $x-1$ entities $e_i$ satisfy 
$\|\zeta_i(t^*) - \zeta_j(t^*)\| \ge \sd[t^*]{x}{j} + 2 $,
these same entities 
have perceived location at time $t^*$ 
at least distance 
$(1 -\lambda_{d,x}/6 )(\sd[t^*]{x}{j} + 2 ) - 2 \lambda_{d,x} \sd[t^*]{x}{j})/6 
= (1 -3\lambda_{d,x}/6 ) \sd[t^*]{x}{j} + 
2(1 -\lambda_{d,x}/6 ) $
from $\zeta_j(p_j(t^*))$.
Thus 
\begin{align*}
\psd[t^*]{x}{j} 
&\ge (1 - 3\lambda_{d,x}/6 ) \sd[t^*]{x}{j} - \lambda_{d,x}  /3 \\
&\ge (1 - 3\lambda_{d,x}/6 ) \sd[t^*]{x}{j} - (1 - \lambda_{d,x})  \sd[t^*]{x}{j} /3 \\
&\ge (2/3 - \lambda_{d,x}/6) \sd[t^*]{x}{j} 
\ge \sd[t^*]{x}{j} /2.
\end{align*}
Taken together these contradict our assumption that
$\sd[t^*]{x}{j}/2 \leq \psd[t^*]{x}{j} \leq 3\sd[t^*]{x}{j}/2$ 
does not hold, and hence 
$\sd[t']{x}{i}/2 \leq \psd[t']{x}{i} \leq 3\sd[t']{x}{i}/2$
at all times $t' \ge t_0$ that entity $e_i$ is queried.

With this it follows by the same argument as above, replacing $t^*$ by an arbitrary $t' \ge p_i(t_0))$, that 
$\sd[t']{x}{i}/2 \leq \psd[t']{x}{i} \leq 3\sd[t']{x}{i}/2$.
\end{proof}

\section{Establishing Perception-Reality Preconditions}\label{app:init}

Prior to performing any queries, our perception of the $x$-separation between entities is far from reality.
So it remains to show that the relationship between perceived and real $x$-separation sufficient to invoke Lemma~\ref{lem:NEWPerception} can be established at some time $t_0$.
One way to achieve this
is to query following a modified version of the FTT$[x\!+\!\Delta]$ scheme of Section~\ref{sec:fixedtarget},
using higher query frequency and a more restrictive criterion than 
$(x\!+\!\Delta)$-degree-safety.

\LemInit*

\begin{proof}
The FTT$[x\!+\!\Delta]$ scheme described in Section~\ref{sec:fixedtarget} is modified as follows.
Instead of conducting each successive query round-robin within half of the time remaining to the target time, we use just a fraction $1/16$ of the time remaining. This means that with each successive round robin phase, the time remaining to the target decreases by a factor $b = 15/16$.

We say that an entity is {\em $(x\!+\!\Delta)$-degree-super-safe} 
at time 
$t_0 - b^s t_0$ (i.e. $b^s t_0$ units before the target time)
if its projected uncertainty region at that time is
separated by distance at least $a b^s t_0$ from 
the projected uncertainty regions of all but at most $x\!+\!\Delta-1$ other entities
(so that its $(x\!+\!\Delta)$-separation at the target time is guaranteed to be at least $a b^s t_0$).
This ensures that when $e_i$ is declared $(x\!+\!\Delta)$-degree-super-safe 
both the true and perceived $(x\!+\!\Delta)$-separation of $e_i$ at the target time are at least $a b^s t_0$ (no matter what further queries are performed).

Assuming that $e_i$ 
is $(x\!+\!\Delta)$-degree-super-safe at time $b^s t_0$ before the target but not at time $b^{s-1} t_0$ before the target,
both the true and perceived $(x\!+\!\Delta)$-separation of $e_i$ at the target time are at most 
$(a/b + 4/b^2) b^s t_0$.
Indeed, at time $b^{s-1} t_0$ before target, the separation of surviving projected uncertainty regions is at most $ab^{s-1} t_0$, and
while the $x$-separation at the target time could be more than this, it cannot be more than $4$ times the radius of any surviving uncertainty region (which is less than $b^{s-2} t_0$) plus $ab^{s-1} t_0$.
Since $(a/b + 4/b^2) b^s t_0 < (16 a/15 +5) b^s t_0$, it follows that 
(i) $\psd[t_0]{x\!+\!\Delta}{i} 
\leq  \frac{16 a/15 +5}{a} \; \sd[t_0]{x\!+\!\Delta}{i}$, and
(ii) $t_0 - p_i(t_0) \leq  \frac{16}{15 a} \psd[p_i(t_0)]{x\!+\!\Delta}{i}$.
Choosing $a$ large enough ($a \ge 64/(5 \lambda_{d,x})$ suffices) guarantees the desired properties.

Following the analysis of the FTT$[x\!+\!\Delta]$ scheme, if the $s$th query round uses $q_s$ queries then any query that guarantees uncertainty degree at most $x$ at time $t_0 + ab^s t_0/2 < (a+1)t_0$ must use at least
$\Theta(\frac{q_s(1\!+\!\Delta)}{x\!+\!\Delta})$ 
queries between time $t_0 - b^{s-1} t_0$, the start of the $s$th query round, and time $t_0 + ab^s t_0/2$;
otherwise some entity that was not $(x\!+\!\Delta)$-degree-super-safe at time $t_0 - b^{s-1} t_0$ would not be $x$-safe at time $t_0 + ab^s t_0/2$.
It follows 
that this initialization scheme uses a minimum query granularity that is competitive to within a factor of 
$\Theta(\frac{x\!+\!\Delta}{1\!+\!\Delta})$
with the minimum granularity used by any other scheme that guarantees the uncertainty degree 
is at most $x$, 
\end{proof}

\section{Proof of Theorem~\ref{thm:BAsuccess}}
\label{sec:basicbucket}

\ThmBAsuccess*

\begin{proof}
It is straightforward to confirm that the assignment of entities to buckets (specified in line 7) ensures that the time between successive queries to any entity $e_i$ satisfies precondition (ii) of Lemma~\ref{lem:NEWPerception}.
From the proof of Lemma~\ref{lem:NEWPerception} we see that this in turn implies that
$t - p_i(t) \le \frac{\lambda_{d,x}}{6} \sd[t]{x}{i}$,
for all entities $e_i$ and all $t \ge t_0$.
But $\frac{\lambda_{d,x}}{6} \sd[t]{x}{i} \le \frac{\lambda_{d, x}}{\lambda_{d, x} +2} \sigma_i(x, t)$, and so 
following the identical analysis used in the proof of Lemma~\ref{lem:FWRRupper}, we conclude that uncertainty degree at most $x$ is maintained indefinitely.

Since no entity has a query scheduled in overlapping buckets, it follows that if the basic BucketScheme$[x]$ makes $3|\E|$ queries over $T$ then, among these, it must make at least $|\E|$ queries to entities in buckets that are fully spanned by $T$. 
Since each entity in each fully spanned bucket contributes $\Theta(1)$ to $\phi_{\E,x}(T)$,
it follows that $\phi_{\E,x}(T) = \Omega(|\E|)$.
\end{proof}

\end{document}